%% file: 2026_ZhRy_quantile.tex
\documentclass[11pt,oneside,fleqn,reqno]{article}
\usepackage{amsmath}
\usepackage{amsthm}
\usepackage{amssymb}
\usepackage{graphicx}
\usepackage{xcolor}
\usepackage{mathtools}
\usepackage{wrapfig}
\usepackage{natbib} \citeindextrue
\usepackage{subfigure}
\usepackage{dsfont}
\usepackage{multirow}
\usepackage{multicol}
\usepackage{algorithm}
\usepackage{algorithmic}
\usepackage{bm}
\usepackage{comment}

\DeclareMathOperator{\conv}{conv}

\allowdisplaybreaks


\input{notation_ams}

\input{macro_ams}

\begin{document}

\title{Quantile optimization in semidiscrete optimal transport}
\author{Yinchu Zhu \and Ilya O. Ryzhov}

\date{\today}

\maketitle

\begin{abstract}
Optimal transport is the problem of designing a joint distribution for two random variables with fixed marginals. In virtually the entire literature on this topic, the objective is to minimize expected cost. This paper is the first to study a variant in which the goal is to minimize a quantile of the cost, rather than the mean. For the semidiscrete setting, where one distribution is continuous and the other is discrete, we derive a complete characterization of the optimal transport plan and develop simulation-based methods to efficiently compute it. One particularly novel aspect of our approach is the efficient computation of a tie-breaking rule that preserves marginal distributions. In the context of geographical partitioning problems, the optimal plan is shown to produce a novel geometric structure.
\end{abstract}

\doublespace

\section{Introduction}\label{sec:intro}

Optimal transport \citep{Vi21} is a rich and challenging branch of optimization focusing on the efficient transformation of one probability distribution into another. Given two random variables, $X$ and $Y$, with known marginal distributions, the goal is to design a \textit{joint} distribution that minimizes some cost function depending on both $X$ and $Y$. In the special case where $X$ and $Y$ are both discrete, optimal transport reduces to the well-known ``transportation problem'' \citep{FoFu56}, in which linear programming is used to match supply to demand. However, when one or both of the random variables are continuous, the problem becomes much more difficult. Applications of these more general settings arise in facility logistics \citep{CaCaDe16}, computer vision \citep{BoDi23}, text mining \citep{Yurochkin19}, portfolio optimization \citep{Nguyen25}, and algorithmic fairness \citep{Chzhen20}. Optimal transport also interfaces with more theoretical areas of operations research, e.g., the literature on distributionally robust optimization, which often uses very similar representations of cost \citep{BlMuZh22,GaKl23}. The purely theoretical literature on optimal transport is surveyed in multiple monographs such as \cite{Sa15}, \cite{PeCu19}, \cite{Vi21}, and \cite{Ma23}.

The objective optimized in this extensive literature is, invariably, an \textit{expected} cost. In other words, if $\pi$ is the joint distribution of $\left(X,Y\right)$, the problem is written as $\inf_{\pi} \mathbb{E}\left[c\left(X,Y\right)\right]$ for some cost function $c$, subject to constraints that preserve the marginal distributions. In fact, the vast majority of the optimal transport literature is focused on \textit{one specific} cost function, namely $c\left(x,y\right) = \|x-y\|^2_2$, in which case the objective becomes the well-known Wasserstein distance metric \citep{CuPe18}. Some theoreticians have examined other choices, such as the $1$-norm \citep{XiLiZhZh24}, the $\infty$-norm \citep{ChDeJu08}, and more generally, cost functions satisfying the ``twist condition'' \citep{McRi16}, a class that includes $p$-norms. Even these special cases are somewhat rarely seen, and in any case, they are also minimized in expectation.

The present paper is, to our knowledge, the first study of optimal transport where the objective function is not the expected cost, but a \textit{quantile} of it, that is,
\begin{equation}\label{eq:obj}
\inf_{\pi} \mathbb{Q}^\alpha_{\pi}\left[c\left(X,Y\right)\right],
\end{equation}
where $\alpha \in \left(0,1\right)$ and
\begin{equation*}
\mathbb{Q}^\alpha_{\pi}\left[c\left(X,Y\right)\right] = \inf\left\{t : P_\pi\left(c\left(X,Y\right) \leq t\right) \geq \alpha\right\}.
\end{equation*}
The rest of the optimal transport framework remains unchanged: we still choose a joint distribution under pre-specified marginals. In general, quantile optimization is used when the underlying distributions are heavy-tailed or skewed, so that the means either do not exist or are heavily affected by extreme values. Such situations can arise in energy \citep{KiPo11}, finance \citep{CaGaMoOl22}, and even traditional OR application areas like inventory management \citep{GlLi96}, where rare events such as supply chain disruptions may yield unusually high costs. One may then prefer to optimize the median or an upper quantile of the cost to produce more risk-averse decisions. With regard to optimal transport in particular, we may consider the problem of geographical partitioning \citep{CaDe13}, where a territory is divided between facilities with given capacities; here, the decision-maker may prefer to design the service zones to, e.g., reduce the incidence of very high travel times.

Our study takes place in the context of \textit{semidiscrete} optimal transport, meaning that $X$ has a general (continuous) distribution, but $Y$ is discrete. While this is not the most general setting, there are compelling reasons to examine it. First, semidiscrete OT is, in and of itself, highly relevant to applications, including the aforementioned geographical partitioning problem, as well as economic choice modeling \citep{BoPe24}, artificial intelligence \citep{Chen19}, data analytics \citep{Gao24}, and computational biology \citep{SePrDa22}. Second, semidiscrete OT is far more flexible regarding the choice of $c$ than the general case, and does not have to be restricted to squared distances or $p$-norms. It is worth noting that, in geographical partitioning, the most commonly used cost function is the \textit{unsquared} Euclidean distance \citep{HaSc20}, which already falls outside much of the OT literature. Third, as \cite{TaShKu22} points out, many numerical algorithms for \textit{general} optimal transport work by discretizing at least one of the distributions anyway; for example, the celebrated Sinkhorn algorithm \citep{Cu13} is based on discretization.

In fact, outside of very narrow special cases, it is exceedingly difficult to compute exact solutions for general OT, and the relevant literature offers mainly theoretical characterizations \citep{GaSw98}. Semidiscrete OT is the exception to the rule. \cite{GeCuPeBa16} was the first to show that, in the classical formulation where the expected cost is minimized, it is possible to learn the exact optimal solution using very simple, fast, simulation-based iterative methods (essentially a one-line stochastic approximation algorithm; see, e.g., \citealp{KuYi03}) without having to discretize $X$. This tractability hinges on the fact that, in the semidiscrete setting, although $\pi$ itself is infinite-dimensional, the optimal solution can be completely characterized using only finitely many parameters. Thus, by studying quantile optimization in semidiscrete OT, we may sacrifice some nominal theoretical generality (partially compensated by the ability to handle a much greater variety of cost functions), but we gain a very clean characterization of the optimal solution that (as we show) can be computed very efficiently.

We provide a comprehensive analysis of the problem, proceeding in several steps:
\begin{enumerate}
\item We begin with a simpler feasibility determination problem, where both $t$ and $\alpha$ are fixed and the question is whether it is possible to find any distribution $\pi$ under which $t$ is the $\alpha$-quantile of $\left(X,Y\right)$. In other words, before setting out to optimize the quantile, we address the more basic issue of whether a given value is feasible as the quantile to begin with. Using an infinite-dimensional analog of Farkas' lemma, we derive a certificate of feasibility. This derivation has some similarity to how Kantorovich duality is typically used in optimal transport theory, but it is not a ``dual'' of the quantile optimization problem, because, thus far, we deal with only fixed $t$ and $\alpha$.
\item Using these results, we fix $t$ and characterize a range of $\alpha$ values that are feasible for that $t$. For quantile optimization, the key step is to link $t$ with the largest $\alpha$ in that range, since we prefer to place as much of the cost distribution below $t$ as possible. We characterize this correspondence as the solution to a nonlinear, but convex and \textit{finite}-dimensional optimization problem (also ``dual-like,'' but not exactly the dual of quantile optimization).
\item In quantile optimization, we start with $\alpha$ instead of $t$. Therefore, the solution is to invert the correspondence obtained in the previous step, i.e., to find $t$ for which $\alpha$ is the largest feasible value. Essentially, this is a stochastic root-finding problem that solves a difficult (but finite-dimensional) convex program in each iteration, in marked contrast with traditional (mean-minimizing) semidiscrete OT, where it is sufficient to solve only one program of a similar type.
\item We then show how to construct the \textit{solution} $\pi$ that achieves the optimal quantile. Like in traditional semidiscrete OT, the value of $Y$ is determined using a criterion computed from an observed value of $X$. Unlike the traditional problem, this criterion often produces ties for different candidate values of $Y$. We show the existence, and characterize the structure, of a tie-breaking rule that preserves the required marginal distribution of $Y$. In the language of the OT literature, the Monge property generally does not hold in quantile optimization, and the tie-breaking rule governs the optimal splitting of mass.
\end{enumerate}
The results of this analysis leave us with two computational challenges: solving the root-finding problem, and computing the tie-breaking rule. Both problems can be approached using simulation, i.e., we generate data from the distribution of $X$ and use these samples to estimate various difficult expectations. We do not even need to know the distribution of $X$; we only require a means of sampling from it. Thus, we may leverage techniques used in the simulation community. As laid out by \cite{KiHe08}, simulation-based methods generally fall into two categories: stochastic approximation (SA, also known as stochastic gradient descent), and sample average approximation (SAA). Both approaches have complementary strengths, and in fact, we have reason to use both in this work. Although stochastic approximation is often used to solve stochastic root-finding problems (even for some types of quantile optimization problems; see, e.g., \citealp{HuPeZhZh22}), in this particular case that problem is coupled with a difficult convex program whose value is not differentiable in the decision variable. Therefore, we instead use SAA to learn the optimal quantile, because it requires us to solve a relatively simple linear program and has strong convergence guarantees. On the other hand, when it comes to learning the tie-breaking rule, SAA is unwieldy because there are too many constraints, while SA turns out to be extremely easy to compute. We give a convergence proof showing that the true optimal transport plan is learned at the canonical $\mathcal{O}\left(N^{-\frac{1}{2}}\right)$ rate, where $N$ is the sample size.

Finally, we present numerical examples illustrating the computational tractability of our approach in the setting of geographical partitioning, and in the process obtain new insights into that application. For these problems, traditional (mean-minimizing) semidiscrete OT is solved by a classical geometric structure known as the additively weighted Voronoi diagram \citep{Au91}. Quantile optimization gives rise to a completely different, partially randomized structure, which, to our knowledge, has never been observed before in the literature on computational geometry. We visualize and contrast both mean- and quantile-optimizing solutions and show their impact on the cost distribution.

To summarize, our paper makes the following contributions. 1) We are the first to optimize a quantile objective for any kind of optimal transport. 2) For the semidiscrete setting, we develop a novel analytical framework that produces a clean characterization of the optimal solution. Two key aspects of the solution, both novel relative to the traditional setting, are that the optimal quantile solves a root-finding problem, and that the optimal $\pi$ requires a probabilistic tie-breaking rule. 3) We show how both aspects may be addressed using simulation-based methods, which are both computationally tractable and provably efficient. In particular, the estimated choice probabilities characterizing the optimal transport plan converge to their true values at the canonical square-root rate. 4) Using these methods, we compute and illustrate a novel geometric structure for the partitioning problem. In this way, our work provides cross-disciplinary insights that, we hope, will be of interest to operations researchers and computational geometers, as well as to the community working on optimal transport.

\section{Modeling and analysis}

Formally, our problem has the objective (\ref{eq:obj}) subject to the constraints
\begin{eqnarray}
\int_{\mathcal{X}} \pi\left(x,k\right)dx &=& p_k, \qquad k = 1,...,K,\label{eq:target}\\
\sum^K_{k=1} \pi\left(x,k\right) &=& f\left(x\right), \qquad x\in\mathcal{X},\label{eq:density}\\
\pi\left(x,k\right) &\geq& 0, \qquad x\in\mathcal{X}, \; k = 1,...,K.\label{eq:nonneg}
\end{eqnarray}
where $\mathcal{X}\subseteq \mathds{R}^d$ is the support of $X$, and $\left\{1,...,K\right\}$ is the support of $Y$. We assume that $c$ satisfies $P\left(c\left(X,k\right)<\infty\right) = 1$ for any $k$. The decision variable in (\ref{eq:obj})-(\ref{eq:nonneg}) is $\pi$, which can be viewed as the mixed joint likelihood of $\left(X,Y\right)$. In other words, $P\left(X\in dx,Y=k\right) = \pi\left(x,k\right)dx$. The density $f$ and the probability vector $p$ (with $p_k > 0$ for all $k$) are fixed problem inputs representing the marginal density of $X$ and the marginal pmf of $Y$, respectively. Constraints (\ref{eq:target})-(\ref{eq:density}) ensure that $\pi$ respects the pre-specified marginals.

Our technical analysis consists of three parts. First, in Section \ref{sec:feasibility}, we formulate a simpler feasibility determination problem and characterize its solutions. Section \ref{sec:quantile} then connects this problem to quantile optimization; essentially, the optimal quantile is related to a particular instance of the feasibility determination problem with appropriately specified inputs. Finally, Section \ref{sec:tiebreaking} characterizes ``tie-breaking rules,'' which are needed to design an implementable optimal transport plan that achieves the desired quantile.

\subsection{First building block: feasibility determination}\label{sec:feasibility}

Before proceeding to quantile optimization, let us first consider a simpler feasibility problem. For \textit{fixed} $t$ and $\alpha$, we wish to determine if there exists a joint distribution $\pi$ satisfying (\ref{eq:target})-(\ref{eq:nonneg}) as well as
\begin{equation}\label{eq:Pequality}
P_\pi\left(c\left(X,Y\right) \leq t\right) = \alpha.
\end{equation}
The answer is not obvious because $\alpha$ and $t$, as well as $f$ and $p$, are arbitrary. If a feasible distribution does not always exist, we wish to characterize the conditions under which feasibility occurs.

We may rewrite (\ref{eq:Pequality}) as
\begin{equation}\label{eq:equality}
\sum_k \int_{\mathcal{X}} 1_{\left\{c\left(x,k\right)\leq t\right\}}\pi\left(x,k\right)dx = \alpha,
\end{equation}
a linear constraint in $\pi$. Thus, (\ref{eq:equality}) together with (\ref{eq:target})-(\ref{eq:density}) is an infinite-dimensional system of linear equations. With the addition of the nonnegativity constraints (\ref{eq:nonneg}), we may apply the following infinite-dimensional analog of Farkas' lemma. The proof is given in the Appendix.

\begin{thm}\label{thm:farkas}
The system described by (\ref{eq:equality}) and (\ref{eq:target})-(\ref{eq:nonneg}) has a solution $\pi$ if and only if there is \textit{no} solution $\left(\zeta,\phi,\psi\right)$ to the system
\begin{equation}\label{eq:dualcons}
1_{\left\{c\left(x,k\right)\leq t\right\}}\zeta + \phi\left(x\right) + \psi_k \geq 0, \qquad \forall x\in\mathcal{X}, \quad k = 1,...,K
\end{equation}
satisfying
\begin{equation}\label{eq:farkas}
\alpha\zeta + \int_{x\in\mathcal{X}} f\left(x\right)\phi\left(x\right)dx + \sum^K_{k=1} p_k \psi_k < 0.
\end{equation}
\end{thm}

Thus, determining the feasibility of the original system is equivalent to solving the infinite-dimensional linear program
\begin{equation}\label{eq:farkasobj}
\inf_{\zeta,\phi,\psi} \alpha\zeta + \int_{x\in\mathcal{X}} f\left(x\right)\phi\left(x\right)dx + \sum^K_{k=1} p_k \psi_k,
\end{equation}
subject to (\ref{eq:dualcons}), and verifying that the optimal value is positive. We may simplify this problem by rewriting (\ref{eq:dualcons}) as
\begin{equation*}
\phi\left(x\right) \geq -1_{\left\{c\left(x,k\right)\leq t\right\}}\zeta - \psi_k, \qquad \forall x\in\mathcal{X}, \quad k = 1,...,K,
\end{equation*}
and setting
\begin{equation}\label{eq:phi}
\phi\left(x\right) = \max_k -1_{\left\{c\left(x,k\right)\leq t\right\}}\zeta - \psi_k
\end{equation}
for every $x\in\mathcal{X}$. Substituting (\ref{eq:phi}) into (\ref{eq:farkasobj}) yields the nonlinear, but finite-dimensional, unconstrained, and convex optimization problem
\begin{equation}\label{eq:dualobj}
\min_{\zeta,\psi} \alpha\zeta + \sum^K_{k=1} p_k \psi_k - \mathbb{E}\left(\min_k 1_{\left\{c\left(X,k\right)\leq t\right\}}\zeta + \psi_k\right).
\end{equation}
As the next result shows, the feasibility of the system (\ref{eq:dualcons})-(\ref{eq:farkas}) is determined entirely by whether (\ref{eq:dualobj}) has an optimal solution.

\begin{prop}\label{prop:unbounded}
Problem (\ref{eq:dualobj}) is unbounded if and only if (\ref{eq:dualcons})-(\ref{eq:farkas}) is feasible. Moreover, if (\ref{eq:dualobj}) is not unbounded, its optimal value is zero.
\end{prop}

\noindent\textbf{Proof:} If (\ref{eq:dualobj}) is unbounded, (\ref{eq:dualcons})-(\ref{eq:farkas}) are feasible by Theorem \ref{thm:farkas}. Suppose now that (\ref{eq:dualobj}) is not unbounded. Then, because this problem is unconstrained and thus cannot be infeasible, it must have an optimal solution $\left(\zeta^*,\psi^*\right)$.

Suppose that the optimal value of (\ref{eq:dualobj}) is strictly negative, i.e.,
\begin{equation*}
\alpha\zeta^* + \sum^K_{k=1} p_k \psi^*_k - \mathbb{E}\left(\min_k 1_{\left\{c\left(X,k\right)\leq t\right\}}\zeta^* + \psi^*_k\right) < 0.
\end{equation*}
The inequality is preserved if we multiply both sides by arbitrarily large $a > 0$. Thus, we can find $a$ for which $\left(a\zeta^*,a\psi^*\right)$ has a better objective value than $\left(\zeta^*,\psi^*\right)$, contradicting the fact that $\left(\zeta^*,\psi^*\right)$ is optimal. Therefore, the optimal value of (\ref{eq:dualobj}) cannot be strictly negative, whence (\ref{eq:dualcons})-(\ref{eq:farkas}) are infeasible by Theorem \ref{thm:farkas}. Finally, because $\zeta = 0$ and $\psi = 0$ are feasible for (\ref{eq:dualobj}), the optimal value of (\ref{eq:dualobj}) must be zero.\qed

Optimal solutions of (\ref{eq:dualobj}), if they exist (they need not be unique), also provide important structural insight into the original feasibility problem, i.e., the solution of (\ref{eq:Pequality}) and (\ref{eq:target})-(\ref{eq:nonneg}). Suppose that some nontrivial $\left(\zeta^*,\psi^*\right)$ is optimal for (\ref{eq:dualobj}), and consider a random variable $Y$ satisfying
\begin{equation}\label{eq:argmin}
Y\left(\omega\right) \in \arg\min_k 1_{\left\{c\left(X\left(\omega\right),k\right)\leq t\right\}}\zeta^* + \psi^*_k
\end{equation}
for all $\omega$. There is some indeterminacy in (\ref{eq:argmin}) because the argmin may not be unique. However, if there exists a (possibly random) tie-breaking rule that causes $Y$ to have the correct marginal distribution, then such a $Y$ will be feasible for the original problem. We formally prove this result as follows.

\begin{prop}\label{prop:howtofeasible}
Suppose that $\left(\zeta^*,\psi^*\right)$ is optimal for (\ref{eq:dualobj}) with $\zeta^*\neq 0$. Suppose also that $Y$ satisfies (\ref{eq:argmin}) as well as $P\left(Y = k\right) = p_k$. Then, the joint distribution $\pi$ of $\left(X,Y\right)$ is feasible for (\ref{eq:Pequality}) and (\ref{eq:target})-(\ref{eq:nonneg}).
\end{prop}

\noindent\textbf{Proof:} Constraints (\ref{eq:target})-(\ref{eq:nonneg}) are satisfied by the properties of $Y$. It remains to show (\ref{eq:Pequality}).

First, by Proposition \ref{prop:unbounded}, the optimal value of (\ref{eq:dualobj}) must be zero. Therefore,
\begin{equation}\label{eq:zerovalue}
\mathbb{E}\left(\min_k 1_{\left\{c\left(X,k\right)\leq t\right\}}\zeta^* + \psi^*_k\right)=\alpha\zeta^* + \sum^K_{k=1} p_k \psi^*_k.
\end{equation}
At the same time, we can derive
\begin{eqnarray}
\mathbb{E}\left(\min_k 1_{\left\{c\left(X,k\right)\leq t\right\}}\zeta^* + \psi^*_k\right) &=& \mathbb{E}\left(1_{\left\{c\left(X,Y\right)\leq t\right\}}\zeta^* + \psi^*_Y\right)\label{eq:Yisargmin}\\
&=& \mathbb{E}\left(\sum_k 1_{\left\{Y=k\right\}}1_{\left\{c\left(X,k\right)\leq t\right\}}\zeta^*\right) + \sum_k p_k\psi^*_k\label{eq:Yties}\\
&=& \zeta^*\mathbb{E}\left(1_{\left\{c\left(X,Y\right)\leq t\right\}}\right) + \sum_k p_k\psi^*_k\nonumber\\
&=& \zeta^* P\left(c\left(X,Y\right)\leq t\right) + \sum_k p_k \psi^*_k,\label{eq:Yderivation}
\end{eqnarray}
where (\ref{eq:Yisargmin}) follows from (\ref{eq:argmin}), and (\ref{eq:Yties}) holds because $P\left(Y=k\right) = p_k$. Setting (\ref{eq:zerovalue}) equal to (\ref{eq:Yderivation}) yields (\ref{eq:Pequality}), as required.\qed

It is worth noting that the lack of a unique argmin in (\ref{eq:argmin}) constitutes a profound difference between our setting and classical optimal transport. Much of the OT literature is focused on situations where $Y$ can be written as a deterministic function of $X$, known as a ``Monge map.'' This is the well-known Kantorovich-Monge equivalence property \citep{Pr07,Be21}. In our setting, however, the mass $\pi\left(x,\cdot\right)$ must be divided between different $\pi\left(x,k\right)$, and the tie-breaking rule provides the correct splitting.

We have not yet said anything about how one may compute a tie-breaking rule satisfying the conditions of Proposition \ref{prop:howtofeasible}. In fact, thus far we have not yet established that such a rule even exists. We will return to these issues later, in Section \ref{sec:tiebreaking}. At the moment, however, we can see how the existence of an optimal solution (\ref{eq:dualobj}) guarantees the feasibility of the original system and allows (potentially) a feasible joint distribution to be constructed.

\subsection{Quantile optimization}\label{sec:quantile}

We now use the insights from Section \ref{sec:feasibility} to develop an approach for quantile optimization. Again, by Theorem \ref{thm:farkas}, the original system described by (\ref{eq:Pequality}) and (\ref{eq:target})-(\ref{eq:nonneg}) is feasible if
\begin{equation}\label{eq:positive}
\alpha\zeta + \sum_k p_k \psi_k - \mathbb{E}\left(\min_k 1_{\left\{c\left(X,k\right)\leq t\right\}}\zeta + \psi_k\right) \geq 0
\end{equation}
for all $\zeta,\psi$. Note that, if $\zeta = 0$, then (\ref{eq:positive}) reduces to $\sum_k p_k\psi_k - \min_k \psi_k \geq 0$, which always holds for all $\psi$. If $\zeta \neq 0$, we may rewrite (\ref{eq:positive}) as
\begin{equation*}
\alpha\cdot\frac{\zeta}{\left|\zeta\right|} + \sum_k p_k \frac{\psi_k}{\left|\zeta\right|} - \mathbb{E}\left(\min_k 1_{\left\{c\left(X,k\right)\leq t\right\}}\frac{\zeta}{\left|\zeta\right|} + \frac{\psi_k}{\left|\zeta\right|}\right) \geq 0.
\end{equation*}
Since we may rescale $\psi$ without changing the inequality, it is sufficient to verify (\ref{eq:positive}) for $\zeta \in\left\{-1,1\right\}$. In other words, for all $\psi$, we must have
\begin{eqnarray*}
\alpha + \sum_k p_k \psi_k - \mathbb{E}\left(\min_k 1_{\left\{c\left(X,k\right)\leq t\right\}} + \psi_k\right) &\geq 0,\\
-\alpha + \sum_k p_k \psi_k - \mathbb{E}\left(\min_k -1_{\left\{c\left(X,k\right)\leq t\right\}} + \psi_k\right) &\geq 0.
\end{eqnarray*}
Therefore,
\begin{equation*}
\mathbb{E}\left(\min_k 1_{\left\{c\left(X,k\right)\leq t\right\}} + \psi_k\right) - \sum_k p_k \psi_k \leq \alpha \leq - \mathbb{E}\left(\min_k -1_{\left\{c\left(X,k\right)\leq t\right\}} + \psi_k\right)+\sum_k p_k \psi_k
\end{equation*}
must hold for all $\psi$. Equivalently, we may write
\begin{equation*}
\max_\psi \mathbb{E}\left(\min_k 1_{\left\{c\left(X,k\right)\leq t\right\}} + \psi_k\right) - \sum_k p_k \psi_k \leq \alpha \leq \min_\psi - \mathbb{E}\left(\min_k -1_{\left\{c\left(X,k\right)\leq t\right\}} + \psi_k\right)+\sum_k p_k \psi_k.
\end{equation*}
Using a change of variables $\psi' = -\psi$ for the upper bound, we have
\begin{equation}\label{eq:bounds}
\max_\psi \mathbb{E}\left(\min_k 1_{\left\{c\left(X,k\right)\leq t\right\}} + \psi_k\right) - \sum_k p_k \psi_k \leq \alpha \leq \min_{\psi'} \mathbb{E}\left(\max_k 1_{\left\{c\left(X,k\right)\leq t\right\}} + \psi'_k\right)-\sum_k p_k \psi'_k.
\end{equation}
By the max-min inequality, the lower bound in (\ref{eq:bounds}) is always less than or equal to the upper bound. Both the maximization problem in the lower bound and the minimization problem in the upper bound are unconstrained, and therefore feasible. Consequently, if we are given a value of $t$, (\ref{eq:bounds}) describes a nonempty interval of $\alpha$ values for which it is possible to make that $t$ the $\alpha$-quantile of the cost.

If our goal is quantile optimization, we should make $\alpha$ as large as possible for fixed $t$. In other words, given a fixed threshold for the cost, we would prefer to have as much of the cost distribution below this threshold as possible. We can connect the upper bound in (\ref{eq:bounds}) back to (\ref{eq:dualobj}) through the following result.

\begin{prop}\label{prop:alphastar}
Given fixed $t$, let
\begin{equation}\label{eq:alphastar}
\alpha^* = \min_{\psi'} \mathbb{E}\left(\max_k 1_{\left\{c\left(X,k\right)\leq t\right\}} + \psi'_k\right)-\sum_k p_k \psi'_k,
\end{equation}
and let $\tilde{\psi}$ be any minimizer of the right-hand side of (\ref{eq:alphastar}). Then, $\zeta^* = -1$ and $\psi^* = -\tilde{\psi}$ optimally solve (\ref{eq:dualobj}) with $\alpha = \alpha^*$.
\end{prop}

\noindent\textbf{Proof:} Problem (\ref{eq:dualobj}) is not unbounded because $\alpha^*$ is, by construction, an element of the interval (\ref{eq:bounds}). Plugging $\left(\zeta^*,\psi^*\right)$ into the objective of (\ref{eq:dualobj}) yields
\begin{eqnarray*}
\alpha^*\zeta^* + \sum^K_{k=1} p_k \psi^*_k - \mathbb{E}\left(\min_k 1_{\left\{c\left(X,k\right)\leq t\right\}}\zeta^* + \psi^*_k\right) &=& -\alpha^* - \sum^K_{k=1} p_k \tilde{\psi}_k - \mathbb{E}\left(\min_k -1_{\left\{c\left(X,k\right)\leq t\right\}} - \tilde{\psi}_k\right)\\
&=& -\alpha^* - \sum^K_{k=1} p_k \tilde{\psi}_k + \mathbb{E}\left(\max_k 1_{\left\{c\left(X,k\right)\leq t\right\}} + \tilde{\psi}_k\right)\\
&=& 0,
\end{eqnarray*}
where the last line is due to the fact that $\tilde{\psi}$ is a minimizer in (\ref{eq:alphastar}). The desired conclusion follows by Proposition \ref{prop:unbounded}.\qed

We also have additional structure on the optimization problem in (\ref{eq:alphastar}), namely, that the right-hand side can always be minimized by binary-valued variables. The proof of this property is deferred to the Appendix so as not to distract from the presentation.

\begin{prop}\label{prop:binary}
There exists a minimizer $\tilde{\psi}$ of the right-hand side of (\ref{eq:alphastar}) satisfying $\tilde{\psi}_k\in\left\{0,1\right\}$ for all $k$.
\end{prop}

We now have a correspondence between $t$ and $\alpha^*\left(t\right)$, the largest feasible quantile for that cost threshold. Proposition \ref{prop:alphastar} then provides us with a corresponding solution of (\ref{eq:dualobj}), and Proposition \ref{prop:howtofeasible} links that solution to a choice of $Y$ (modulo a suitable tie-breaking rule, which still remains to be discussed) that solves (\ref{eq:Pequality}) and (\ref{eq:target})-(\ref{eq:nonneg}) for $t$ and $\alpha^*\left(t\right)$.

The quantile optimization problem that we originally set out to solve is then readily seen to be an inversion of this correspondence: given $\alpha$, we search for $t$ such that $\alpha^*\left(t\right) = \alpha$. It is obvious that $\alpha^*\left(t\right)$ is increasing in $t$, and with some mild regularity conditions, it becomes continuous and strictly increasing, so a suitable $t$ always exists and is unique. The following results examine some simple conditions under which these properties hold; the proofs are deferred to the Appendix.

\begin{prop}\label{prop:continuity}
Suppose that, for every $k$, the random variable $c\left(X,k\right)$ has a density. Then, $\alpha^*\left(t\right)$ is continuous in $t$. Moreover, if the density of each $c\left(X,k\right)$ is bounded above, then $\alpha^*\left(t\right)$ is Lipschitz.
\end{prop}

\begin{prop}\label{prop:monotonicity}
Let $t_1 < t_2$ and suppose that the random vector $\left(c\left(X,1\right),...,c\left(X,K\right)\right)$ has a density that is bounded below by some $C>0$ in a neighborhood of $\left(t_1,...,t_1\right)\in\mathbb{R}^K$. Then, $\alpha^*\left(t_1\right) < \alpha^*\left(t_2\right)$.
\end{prop}

\subsection{Characterization of tie-breaking rules}\label{sec:tiebreaking}

In the following, suppose that we already have the desired value of $t$ (the corresponding $\alpha^*$ is simply the target quantile $\alpha$). Thus, we have the optimal value of the quantile optimization problem (\ref{eq:obj}). For practical purposes, however, we wish to also know the optimal \textit{solution}, i.e., the joint distribution $\pi$ of $\left(X,Y\right)$ that achieves the infimum.

We abuse notation slightly by letting $\psi^*_k$ be a minimizer of the right-hand side of (\ref{eq:alphastar}); from Proposition \ref{prop:alphastar}, we know that any such minimizer can be converted into an optimal solution of (\ref{eq:dualobj}), so it is no longer necessary to distinguish between them in the notation. By combining Propositions \ref{prop:howtofeasible} and \ref{prop:alphastar}, we know that a suitable $\pi$ can be obtained from any $Y$ that satisfies
\begin{equation}\label{eq:argmax}
Y\left(\omega\right) \in \arg\max_k 1_{\left\{c\left(X\left(\omega\right),k\right)\leq t\right\}} + \psi^*_k,
\end{equation}
with ties in the argmax broken in a (possibly randomized) manner that satisfies $P\left(Y = k\right) = p_k$. We will now prove that a randomized tie-breaking rule satisfying these conditions always exists, and give a characterization that will later enable its tractable computation.

First, let us observe that, by definition,
\begin{equation*}
\mathbb{E}\left(\max_k 1_{\left\{c\left(X,k\right)\leq t\right\}} + \psi^*_k\right)-p^\top \psi^* \leq \mathbb{E}\left(\max_k 1_{\left\{c\left(X,k\right)\leq t\right\}} + \psi_k\right)-p^\top \psi, \qquad \forall\psi,
\end{equation*}
which may be rewritten as
\begin{equation*}
\mathbb{E}\left(\max_k 1_{\left\{c\left(X,k\right)\leq t\right\}}+ \psi_k\right) - \mathbb{E}\left(\max_k 1_{\left\{c\left(X,k\right)\leq t\right\}}+ \psi^*_k\right) \geq p^\top\left(\psi-\psi^*\right), \qquad \forall\psi.
\end{equation*}
This means precisely that $\psi^*$ is optimal if and only if $p$ is a subgradient of the function
\begin{equation*}
h\left(\psi\right) = \mathbb{E}\left(\max_k 1_{\left\{c\left(X,k\right)\leq t\right\}} + \psi_k\right)
\end{equation*}
at $\psi^*$. In order to eventually establish the existence of an optimal tie-breaking rule, it is necessary to describe all of the subgradients of $h$.

For notational convenience, let $H_k\left(\psi,X\right) = 1_{\left\{c\left(X,k\right)\leq t\right\}} + \psi_k$ and $H\left(\psi,X\right) = \max_k H_k\left(\psi,X\right)$. Then, $h\left(\psi\right) = \mathbb{E}\left(H\left(\psi,X\right)\right)$. Let us first focus on the subgradients of $H\left(\psi,x\right)$, holding $x$ fixed, and then pass to the expectation. We will also use the notation
\begin{equation*}
S\left(\psi,x\right) = \arg\max_k H_k\left(\psi,x\right)
\end{equation*}
so that we do not have to keep writing out the argmax.

For any $k$ and $x$, $H_k\left(\psi,x\right)$ is differentiable with $\nabla_\psi H_k\left(\psi,x\right) = e_k$, where $e_k$ is a vector of zeroes with the $k$th component equal to $1$. By Lemma 3.1.10 of \cite{Ne04}, the subdifferential $\partial H\left(\psi,x\right)$ of $H\left(\psi,x\right)$ at $\psi$ is given by
\begin{eqnarray*}
\partial H\left(\psi,x\right) &=& \conv\left\{e_j: j \in S\left(\psi,x\right)\right\}\\
&=& \left\{v \in \mathbb{R}^K_+:\sum_k v_k = 1, \, v_k = 0 \mbox{ for } k \notin S\left(\psi,x\right)\right\}.
\end{eqnarray*}
In words, $\partial H\left(\psi,x\right)$ is the probability simplex on only those coordinates $k$ that are elements of the argmax. To put it differently, $v \in \partial H\left(\psi,x\right)$ represents a discrete probability distribution on the elements of $S\left(\psi,x\right)$. Note that the subdifferential only depends on $x$ through $S\left(\psi,x\right)$.

By the main theorem of \cite{RoWe82}, the subdifferential of $h$ is given by $\partial h\left(\psi\right) = \mathbb{E}\left(\partial H\left(\psi,X\right)\right)$. Here, we are referring to the Aumann expectation \citep{Au65} of a random set $\partial H\left(\psi,X\right)$ whose realization is determined by $S\left(\psi,X\right)$; in other words, the expectation may be taken over the discrete distribution of $S\left(\psi,X\right)$, rather than the continuous distribution of $X$. Formally, $w \in \partial h\left(\psi\right)$ if and only if we can write
\begin{equation}\label{eq:wexp}
w = \sum_{S \subseteq \left\{1,...,K\right\}} Q\left(\psi,S\right) v_S,
\end{equation}
where
\begin{equation*}
Q\left(\psi,S\right) = P\left(S = \arg\max_k 1_{\left\{c\left(X,k\right)\leq t\right\}} + \psi_k\right),
\end{equation*}
and $v\left(S\right) \in \mathbb{R}^K_+$ satisfies $\sum_k v_k\left(S\right) = 1$, $v_k = 0$ for $k \notin S$. We may restrict the summation set in (\ref{eq:wexp}) to only those $S$ for which $Q\left(\psi,S\right) > 0$.

Now recall from before that $p \in \partial h\left(\psi^*\right)$. Therefore, by (\ref{eq:wexp}), for every $S$ with $Q\left(\psi^*,S\right) > 0$, we may find $v_S$ such that $p = \sum_S Q\left(\psi,S\right) v_S$. Each vector $v_S$ is a probability distribution, and this distribution is precisely our randomized tie-breaking rule: we first observe $X$ and therefore $S\left(\psi^*,X\right)$, and then let $P\left(Y = k\mid S\left(\psi^*,X\right)\right) = v_{S\left(\psi^*,X\right),k}$. Thus, (\ref{eq:argmax}) is satisfied because $v_{S\left(\psi^*,X\right),k} > 0$ only for $k \in S\left(\psi^*,X\right)$, and
\begin{equation*}
P\left(Y = k\right) = \mathbb{E}\left(P\left(Y = k\mid S\left(\psi^*,X\right)\right)\right) = \sum_S Q\left(\psi^*,S\right) P\left(Y = k\mid S\left(\psi^*,X\right) = S\right) = p_k,
\end{equation*}
where the last equality follows by (\ref{eq:wexp}).

We have now established that a suitable tie-breaking rule always exists. However, in order to compute it, we must find a probability distribution $v_S$ for each $S$ such that $p = \sum_S Q\left(\psi^*,S\right) v_S$. This is difficult, in part because the number of distributions to be computed is combinatorially large, and in part because the probabilities $Q\left(\psi^*,S\right)$ are computationally challenging; however, in Section \ref{sec:comptie}, we will present a computationally tractable scheme that does not require knowledge of $Q$ or enumeration of the possible $S$.

\section{Algorithms and performance guarantees}\label{sec:comp}

Our problem poses two algorithmic challenges. The first is to find the root $\alpha^*\left(t\right) = \alpha$ for given $\alpha$, as described in Section \ref{sec:quantile}, and the second is to compute the tie-breaking rule using the characterization in Section \ref{sec:tiebreaking}. Both involve the computation of difficult integrals: in the first case, we must deal with a difficult expectation in (\ref{eq:alphastar}), and in the second case, we must deal with a combinatorially large number of difficult probabilities in (\ref{eq:wexp}).

As discussed in Section \ref{sec:intro}, we use a simulation-based approach where independent data $\left\{X^n\right\}^N_{n=1}$ are generated from the distribution of $X$. We will use sample average approximation (SAA) to solve the root-finding problem and learn the optimal quantile, and we will use stochastic approximation (SA) to learn the tie-breaking rule. These algorithms are presented and studied in Sections \ref{sec:compquantile} and \ref{sec:comptie}, respectively.

\subsection{Learning the optimal quantile}\label{sec:compquantile}

As discussed in Section \ref{sec:quantile}, our task is to solve $\alpha^*\left(t\right) = \alpha$ given some prespecified $\alpha$. Because the computation of $\alpha^*\left(t\right)$ involves a difficult expectation, we can view this as a stochastic root-finding problem \citep{PaKi11}. Such problems are often solved using the stochastic approximation (SA) method \citep{KuYi03}. SA has previously been applied \citep{GeCuPeBa16,ZhRy26} to solve semidiscrete optimal transport problems with the standard objective of minimizing expected costs. However, in the setting of quantile optimization, we have two ``layers'' of root-finding rather than just one, because $\alpha^*\left(t\right)$ for fixed $t$ is itself the optimal value of a stochastic optimization problem. While variants of SA do exist for such settings \citep{Bo97}, their theoretical guarantees \citep{Do22} are weaker, or require much stronger assumptions, than what is available for standard SA. Moreover, our case has the additional complication that, in (\ref{eq:alphastar}), the expression inside the expected value is not pathwise differentiable due to the possibility of ties between different indices $k$.

For all of these reasons, we do not attempt to learn the optimal quantile through SA. Instead, we use the sample average approximation (SAA) method \citep{KiPaHe14,ShDeRu21}, which replaces the expectation in (\ref{eq:alphastar}) by its empirical equivalent given a sample $\left\{X^n\right\}^N_{n=1}$. The optimization problem over $\psi$ can then be reformulated as a linear program by introducing additional decision variables $z_n$, $n=1,...,N$. For fixed $t$, we solve
\begin{equation*}
\hat{\alpha}^N\left(t\right) = \min_{\psi_k,z_n} \frac{1}{N}\sum^N_{n=1} z_n - \sum_k p_k \psi_k
\end{equation*}
subject to
\begin{equation*}
z_n \geq 1_{\left\{c\left(X^n,k\right)\leq t\right\}} + \psi_k, \qquad n=1,...,N, \; k = 1,...,K.
\end{equation*}
This LP has $K+N$ decision variables and $K\cdot N$ constraints, and can be tractable for a fairly large sample size due to the availability of large-scale LP solvers.

However, we must solve a problem like this for every $t$ value that we might consider. Since $\psi$ is bounded by Proposition \ref{prop:binary} and $\hat{\alpha}^N$ only involves an expectation over a finite distribution, the optimal value is bounded. Then, since $\hat{\alpha}^N\left(t\right)$ is increasing in $t$, we can solve $\hat{\alpha}^N\left(t\right) = \alpha$ through the bisection method. Because we are using a sample average instead of the true expectation, $\hat{\alpha}^N\left(t\right)$ may have jumps and a root may not exist; in this case, the procedure will converge to the value
\begin{equation}\label{eq:hatroot}
\hat{t}^N = \inf\left\{t:\hat{\alpha}^N\left(t\right) \geq \alpha\right\}.
\end{equation}
Note that every iteration of the bisection method works with the same sample, i.e., we may approximate $\hat{t}^N$ to arbitrary precision without increasing $N$. The consistency of $\hat{\alpha}^N$ can be obtained through standard techniques \citep{ShDeRu21}. However, we can go further and characterize the convergence rates of $\hat{\alpha}^N$ and $\hat{t}^N$ under the same assumptions as in Propositions \ref{prop:continuity}-\ref{prop:monotonicity}. The proofs of the following results can be found in the Appendix.

\begin{thm}\label{thm:conc}
Suppose that, for every $k$, the random variable $c\left(X,k\right)$ has a density. Then, there exist $C_1,C_2 > 0$ such that
\begin{equation*}
P\left(\left|\alpha^*\left(\hat{t}^N\right) - \alpha\right| > \frac{w}{\sqrt{N}}\right) \leq C_1 e^{-C_2 w^2}
\end{equation*}
for all $w > 0$.
\end{thm}

\begin{thm}\label{thm:conct}
Suppose that, for every $k$, the random variable $c\left(X,k\right)$ has a density. Suppose also that $t^*$ satisfies $\alpha^*\left(t^*\right) = \alpha$ and that the random vector $\left(c\left(X,1\right),...,c\left(X,K\right)\right)$ has a density that is bounded below by some $C>0$ in a neighborhood of $\left(t^*,...,t^*\right)\in\mathbb{R}^K$. Then, there exist $C_3,C_4 > 0$ such that
\begin{equation*}
P\left(\left|\hat{t}^N - t^*\right| > \delta\right) \leq C_3 e^{-C_4 \delta^2 N}
\end{equation*}
for all sufficiently small $\delta > 0$.
\end{thm}

A consequence of Theorem \ref{thm:conct} is that, with high probability, the SAA solution will give us $\psi$ that is exact for the estimated root $\hat{t}^N$, in the sense of minimizing the right-hand side of (\ref{eq:alphastar}) with $t = \hat{t}^N$. In other words, all of the SAA estimation error will come from $\hat{t}^N$, not the $\psi$ value obtained from the LP.

\begin{thm}\label{thm:decision}
Define the sets
\begin{eqnarray}
\Psi^*\left(t\right) &=& \arg\min_{\psi\in\left\{0,1\right\}^K} \mathbb{E}\left(\max_k 1_{\left\{c\left(X,k\right)\leq t\right\}} + \psi_k\right) - p^\top \psi,\label{eq:Psistar}\\
\hat{\Psi}^N\left(t\right) &=& \arg\min_{\psi\in\left\{0,1\right\}^K} \frac{1}{N}\sum^N_{n=1}\left(\max_k 1_{\left\{c\left(X^n,k\right)\leq t\right\}} + \psi_k\right) - p^\top \psi.\label{eq:Psihat}
\end{eqnarray}
Suppose that we are in the situation of Theorem \ref{thm:conct}. Then, there exist $C_5,C_6 > 0$ such that
\begin{equation*}
P\left(\hat{\Psi}^N\left(\hat{t}^N\right) \neq \Psi^*\left(\hat{t}^N\right)\right) \leq C_5 e^{-C_6 N}.
\end{equation*}
\end{thm}

Recall from (\ref{eq:argmax}) that the joint distribution $\pi$ is designed by computing $Y$ according to a criterion that depends on both $t$ and $\psi$. If $t=\hat{t}^N$, but $\psi\in\Psi^*\left(\hat{t}^N\right)$ with a proper tie-breaking rule, we will obtain a joint distribution that is feasible and has $\hat{t}^N$ as its $\alpha^*\left(\hat{t}^N\right)$-quantile. The only source of error will then be $\hat{t}^N$. Theorem \ref{thm:decision} tells us that we find $\psi$ that is ``correct'' for the estimated quantile $\hat{t}^N$ with high probability.

If we slightly strengthen the conditions of Theorem \ref{thm:decision} (essentially making $\alpha^*\left(t\right)$ Lipschitz as per Proposition \ref{prop:continuity}), we can further show that $\hat{\Psi}^N\left(\hat{t}^N\right) \subseteq \Psi^*\left(t^*\right)$ with high probability. We will use this result in Section \ref{sec:comptie} to assess how close our estimated $\pi$ is to the optimal transport plan.

\begin{thm}\label{thm:subset}
Let $\Psi^*$ and $\hat{\Psi}^N$ be as in (\ref{eq:Psistar})-(\ref{eq:Psihat}). Suppose that, for every $k$, the random variable $c\left(X,k\right)$ has a density that is bounded above. Suppose also that $t^*$ satisfies $\alpha^*\left(t^*\right) = \alpha$ and that the random vector $\left(c\left(X,1\right),...,c\left(X,K\right)\right)$ has a density that is bounded below by some $C>0$ in a neighborhood of $\left(t^*,...,t^*\right)\in\mathbb{R}^K$. Then, there exist $C_7,C_8 > 0$ such that
\begin{equation*}
P\left(\hat{\Psi}^N\left(\hat{t}^N\right) \subseteq \Psi^*\left(t^*\right)\right) \geq 1 - C_7 e^{-C_8 N}.
\end{equation*}
\end{thm}

\subsection{Learning the tie-breaking rule}\label{sec:comptie}

For the moment, let us suppose, as in Section \ref{sec:tiebreaking}, that the desired $t^*$, and its corresponding $\psi^*$, have already been found. In practice, $t^*$ would be replaced by the SAA estimate $\hat{t}^N$. However, from Theorem \ref{thm:decision}, we would still have the correct $\psi^*$ value for $\hat{t}^N$ with high probability. Thus, (\ref{eq:argmax}) combined with a suitable tie-breaking rule will produce a cost distribution whose $\alpha^*\left(\hat{t}^N\right)$-quantile is $\hat{t}^N$. The developments below will produce a rule that achieves $\alpha^*\left(t\right)$ for any $t$ as long as $\psi^*$ minimizes (\ref{eq:alphastar}) for that $t$.

From Section \ref{sec:tiebreaking}, we know that a valid tie-breaking rule is represented by a collection $\left\{v_S\right\}$ of probability distributions. Each $S\subseteq\left\{1,...,K\right\}$ satisfies $Q\left(\psi^*,S\right) > 0$ and each $v_S \in \mathbb{R}^K_+$ satisfies $\sum_k v_{S,k} = 1$ and $v_{S,k} = 0$ for $k \notin S$. Thus, any solution of the system
\begin{eqnarray}
\sum_{k\in S} v_{S,k} &=& 1, \qquad \forall S \mbox{ with } Q\left(\psi^*,S\right) > 0,\label{eq:Sconstraints}\\
\sum_{S:k\in S} Q\left(\psi^*,S\right) v_{S,k} &=& p_k, \qquad k = 1,...,K,\label{eq:kconstraints}
\end{eqnarray}
is a valid tie-breaking rule. It is crucial that \textit{any} solution will yield an optimal transport plan through Proposition \ref{prop:howtofeasible}. Thus, we are free to select any convenient objective function and optimize it under the constraints (\ref{eq:Sconstraints})-(\ref{eq:kconstraints}). The main difficulty is that the number of constraints is combinatorially large, making it prohibitive to apply SAA as in Section \ref{sec:compquantile}.

We propose to use the expected entropy objective
\begin{equation}\label{eq:entropy}
\min_{v_{S,k}} \sum_S Q\left(\psi^*,S\right) \sum_{k\in S} v_{S,k}\log\left(v_{S,k}\right).
\end{equation}
This objective is strictly convex and thus (\ref{eq:entropy}) will have a unique optimal solution. Under this choice, we obtain an unconstrained dual problem with only $K$ decision variables, as shown in the following result.

\begin{prop}\label{prop:entropydual}
Problem (\ref{eq:entropy}) subject to (\ref{eq:Sconstraints})-(\ref{eq:kconstraints}) has the unconstrained dual
\begin{equation}\label{eq:entropydual}
\max_{\theta_k} p^\top \theta - \sum_S Q\left(\psi^*,S\right) \log \sum_{k\in S} e^{\theta_k}.
\end{equation}
\end{prop}

\noindent\textbf{Proof:} By combining Example 3.21, Sec. 3.3.2, and eq. (5.11) in \cite{BoVa04}, we obtain the dual function
\begin{equation}\label{eq:dualfunction}
g\left(\lambda,w,\theta\right) = -\sum_S w_S - \sum_k p_k \theta_k - \sum_S \sum_{k\in S} Q\left(\psi^*,S\right)e^{\frac{\lambda_{S,k}-w_S-Q\left(\psi^*,S\right)\theta_k}{Q\left(\psi^*,S\right)}-1}.
\end{equation}
By eq. (5.16) in \cite{BoVa04}, the dual problem maximizes (\ref{eq:dualfunction}) subject to $\lambda\geq 0$ and no constraints on $w,\theta$. It is optimal to set $\lambda \equiv 0$ and $w_S = Q\left(\psi^*,S\right)\log\left(\sum_{k\in S} e^{-\theta_k -1}\right)$. Plugging these values into (\ref{eq:dualfunction}) leads to (\ref{eq:entropydual}).\qed

Moreover, an optimal solution of the dual can be immediately converted into a valid $v_S$ for any $S$. For implementation purposes, one can do this ``on demand,'' storing only the dual variables and creating a probability distribution only after $X$ and $S\left(\psi^*,X\right)$ have been observed.

\begin{prop}\label{prop:softmax}
Let $\theta^*$ be an optimal solution of (\ref{eq:entropydual}). Then, $v^*_{S,k} = \frac{e^{\theta^*_k}}{\sum_{j\in S}e^{\theta^*_j}}$ for $k\in S$ is optimal for problem (\ref{eq:entropy}) subject to (\ref{eq:Sconstraints})-(\ref{eq:kconstraints}).
\end{prop}

\noindent\textbf{Proof:} We first verify feasibility. Clearly $\sum_k v^*_{S,k} = 1$ for all $S$. The first-order optimality conditions of (\ref{eq:entropydual}) are given by
\begin{equation}\label{eq:entropykkt}
p_k - \sum_S Q\left(\psi^*,S\right) \frac{e^{\theta^*_k}}{\sum_{j\in S} e^{\theta^*_j}}1_{\left\{k\in S\right\}} = 0,
\end{equation}
which is exactly (\ref{eq:kconstraints}).

Lastly, we plug our choice of $v^*$ into the primal objective and find
\begin{eqnarray*}
\sum_S Q\left(\psi^*,S\right) \sum_{k\in S} v^*_{S,k}\log\left(v^*_{S,k}\right) &=& \sum_S Q\left(\psi^*,S\right) \sum_{k\in S} \frac{e^{\theta^*_k}}{\sum_{j\in S}e^{\theta^*_j}} \left(\theta^*_k - \log\left(\sum_{j\in S}e^{\theta^*_j}\right)\right)\\
&=& \sum_S \sum_{k\in S} Q\left(\psi^*,S\right) \frac{e^{\theta^*_k}}{\sum_{j\in S}e^{\theta^*_j}}\theta^*_k - \sum_S Q\left(\psi^*,S\right) \log\left(\sum_{j\in S}e^{\theta^*_j}\right)\\
&=& \sum_k \theta^*_k \sum_{S:k\in S} Q\left(\psi^*,S\right) \frac{e^{\theta^*_k}}{\sum_{j\in S}e^{\theta^*_j}} - \sum_S Q\left(\psi^*,S\right) \log\left(\sum_{j\in S}e^{\theta^*_j}\right)\\
&=& p^\top \theta^* - \sum_S Q\left(\psi^*,S\right) \log\left(\sum_{j\in S}e^{\theta^*_j}\right),
\end{eqnarray*}
where the last line follows by (\ref{eq:entropykkt}). So, strong duality holds for the pair $\left(v^*,\theta^*\right)$, which completes the proof.\qed

Finally, we explain how to optimize the dual. Problem (\ref{eq:entropydual}) can be written as $\max_\theta p^\top \theta - \mathbb{E}\left(\log \sum_{k\in S\left(X\right)} e^{\theta_k}\right)$. The expression inside the expected value is pathwise differentiable (unlike in Section \ref{sec:compquantile}), giving rise to a simple stochastic approximation algorithm
\begin{equation}\label{eq:sa}
\theta^{m+1}_k = \theta^m_k + \gamma_m\left(p_k - \frac{e^{\theta^m_k}}{\sum_{j\in S\left(\psi^*,X^{m+1}\right)} e^{\theta^m_j}}1_{\left\{k\in S\left(\psi^*,X^{m+1}\right)\right\}}\right),
\end{equation}
where $\gamma_m$ is a stepsize and $X^{m+1}$ is an independent sample from the distribution of $X$. We use $m$ instead of $n$ in the superscript because, in a standard implementation of SA, a new independent sample (separate from the ones that were used to estimate $\hat{t}^N$) would be required in every iteration to ensure the validity of the procedure. However, this algorithm does not require us to know $Q$, to enumerate $S$, or to store $v_S$ for every $S$. The initial iterate $\theta^0$ can be arbitrary. A common variation known as Polyak averaging \citep{PoJu92}, which often leads to improved convergence rates, adds an auxiliary iterate
\begin{equation*}
\bar{\theta}^{m+1}_k = \frac{m}{m+1}\theta^m_k + \frac{1}{m+1}\theta^{m+1}_k,
\end{equation*}
which simply averages the iterates obtained through (\ref{eq:sa}). We then use $\bar{\theta}^m$ as our estimates of the dual variables. Whether or not Polyak averaging is used, the computational and storage cost of stochastic approximation is trivial.

There is an extensive theory on the convergence rates of SA; see, e.g., the seminal work of \cite{BaMo11,BaMo13}. There are, however, some issues with applying these results to the sequence $\left\{\theta^m\right\}$ obtained using (\ref{eq:sa}). First, the theory often makes strong assumptions, such as strong convexity, on the objective function being optimized. When these assumptions are not met, one typically obtains a convergence rate for the value of the objective function (\ref{eq:entropydual}), not the decision variable. Second, and perhaps more importantly, the theory would require us to have a second sample from the distribution of $X$, independent of the sample $\left\{X^n\right\}^N_{n=1}$ used to compute $\hat{t}^N$. This would increase the burden of data collection.

For these reasons, we derive a performance guarantee in the following framework. First, since we will be working with SAA estimates of $t^*$ and $\psi^*$, we will abuse notation slightly to make the dependence of the choice on these parameters explicit, by writing $S\left(t,\psi,x\right) = \arg\max_k 1_{\left\{c\left(X,k\right)\leq t\right\}} + \psi_k$. Now suppose that, in (\ref{eq:sa}), we generate $X^{m+1}$ not by sampling from the density $f$, but by \textit{bootstrapping} from the original sample $\left\{X^n\right\}^N_{n=1}$. The set $S\left(\psi^*,X^{m+1}\right)$ in (\ref{eq:sa}) will be replaced by $S\left(\hat{t}^N,\hat{\psi}^N,X^{m+1}\right)$. Then, (\ref{eq:sa}) will optimize, not (\ref{eq:entropydual}), but the sample analog
\begin{equation}\label{eq:entropysample}
\max_\theta L^N\left(\theta,\hat{t}^N,\hat{\psi}^N\right), \qquad L^N\left(\theta,t,\psi\right) = p^\top \theta - \frac{1}{N} \sum^N_{n=1} \log \left(\sum_{k\in S\left(t,\psi,X^{n+1}\right)} e^{\theta_k}\right).
\end{equation}
There are two advantages to doing this. First, we may work with one sample of size $N$ for all of our computations, without having to divide the total data-collection effort into two samples. Second, the computational cost of running SA becomes separated from $N$, because (\ref{eq:sa}) is very fast and can easily be run for much more than $N$ iterations. Thus, we may assume that the optimal solution of (\ref{eq:entropysample}) can be cheaply computed to arbitrary precision, and the only remaining task is to characterize the convergence of the optimal solution $\hat{\theta}^N$ of (\ref{eq:entropysample}), and the corresponding optimal transport plan, to the true optimal solutions.

First, we build on Theorem \ref{thm:subset} to characterize the probability that the estimated $\hat{t}^N$ leads to an error in the choice set (before tie-breaking). Note the distinction: Theorem \ref{thm:subset} studies the set of parameters characterizing the optimal assignment rule, whereas Proposition \ref{prop:newargmax} below concerns a specific assignment of a \textit{new} $X$ not in the sample. In other words, we have moved from assessing the accuracy of the estimated parameters to evaluating their quality in making decisions for other $X$ that had not previously been observed. The probability that the choice set estimated for that $X$ under $\hat{t}^N$ and $\hat{\psi}^N$ does not match the true choice set (computed using $t^*$ and $\psi^*$) is shown to be of the canonical order $\mathcal{O}\left(N^{-\frac{1}{2}}\right)$.

\begin{prop}\label{prop:newargmax}
Suppose that we are in the situation of Theorem \ref{thm:subset}. Let $\psi^* \in \Psi^*\left(t^*\right)$ and $\hat{\psi}^N \in \hat{\Psi}^N\left(\hat{t}^N\right)$, and let $X$ be a new draw from the density $f$, independent of $\left\{X^n\right\}^N_{n=1}$. Then,
\begin{equation*}
P\left(S\left(\hat{t}^N,\hat{\psi}^N,X\right) \neq S\left(t^*,\psi^*,X\right)\right) \lesssim N^{-\frac{1}{2}}.
\end{equation*}
\end{prop}

Next, we study the accuracy of the solution of (\ref{eq:entropysample}). Since adding the same constant to each component $\theta$ will not change the value of the objective function, we may normalize $\theta_K = 0$ and think of both (\ref{eq:entropydual}) and (\ref{eq:entropysample}) as problems in $K-1$ dimensions. It is then fairly straightforward to show that (\ref{eq:entropydual}) has a unique solution $\theta^*$. The gradient of (\ref{eq:entropysample}) is given by
\begin{equation}\label{eq:samplegradient}
\ell^N_k\left(\theta,t,\psi\right) = \frac{\partial L^N\left(\theta,t,\psi\right)}{\partial\theta_k} = p_k -\frac{1}{N} \sum^N_{n=1}\frac{e^{\theta^n_k}}{\sum_{j\in S\left(t,\psi,X^{n+1}\right)} e^{\theta^m_j}}1_{\left\{k\in S\left(t,\psi,X^{n+1}\right)\right\}}.
\end{equation}
The next result establishes that, when (\ref{eq:samplegradient}) is sufficiently small, we have an accurate approximation of $\theta^*$.

\begin{thm}\label{thm:gradient}
Suppose that we are in the situation of Theorem \ref{thm:subset}, and let $\hat{\psi}^N \in \hat{\Psi}^N\left(\hat{t}^N\right)$. Additionally, suppose that $\hat{\theta}$ satisfies
\begin{equation}\label{eq:gradientbound}
\mathbb{E}\left(\|\ell^N\left(\hat{\theta},\hat{t}^N,\hat{\psi}^N\right)\|_2\right) \lesssim N^{-\frac{1}{2}}.
\end{equation}
Then, $\mathbb{E}\left(\|\hat{\theta}-\theta^*\|_2\right) \lesssim N^{-\frac{1}{2}}$.
\end{thm}

Note that $\hat{\theta}$ in (\ref{eq:gradientbound}) need not be deterministic. In fact, if $\hat{\theta} = \hat{\theta}^N$, i.e., we have the exact optimal solution of (\ref{eq:entropysample}), the condition holds trivially since the left-hand side of (\ref{eq:gradientbound}) is equal to zero. We allow some small amount of error, the idea being that SA can always be run long enough to make (\ref{eq:gradientbound}) hold.

Our final result bounds the error of the \textit{choice probabilities} computed from $\hat{t}^N$, $\hat{\psi}^N$, and $\hat{\theta}^N$ (or a close approximation), again for a new data point $X$. This completes our analysis, as the choice probabilities are precisely what characterizes the optimal transport plan for an arbitrary $X$. We see that the canonical $\mathcal{O}\left(N^{-\frac{1}{2}}\right)$ rate is obtained.

\begin{thm}\label{thm:final}
Define
\begin{equation*}
\rho_k\left(x,\theta,t,\psi\right) = \frac{e^{\theta_k}1_{\left\{k\in S\left(x,t,\psi\right)\right\}}}{1_{\left\{K\in S\left(x,t,\psi\right)\right\}+\sum_{j<K-1} e^{\theta}_j1_{\left\{j\in S\left(x,t,\psi\right)\right\}}}}.
\end{equation*}
Let $\psi^* \in \Psi^*\left(t^*\right)$ and $\hat{\psi}^N \in \hat{\Psi}^N\left(\hat{t}^N\right)$, and let $X$ be a new draw from the density $f$, independent of $\left\{X^n\right\}^N_{n=1}$. Suppose that we are in the situation of Theorem \ref{thm:subset}. Additionally, suppose that $\hat{\theta}$ satisfies (\ref{eq:gradientbound}). Then, $\mathbb{E}\left(\left|\rho_k\left(X,\hat{\theta},\hat{t}^N,\hat{\psi}^N\right)-\rho_k\left(X,\theta^*,t^*,\psi^*\right)\right|\right) \lesssim N^{-\frac{1}{2}}$.
\end{thm}

\section{Application: geographical partitioning}

Geographical partitioning problems are a major application area of semidiscrete optimal transport \citep{CaCaDe16,CaPeRy24}. In such problems, we consider a geographical region $\mathcal{X}\subseteq\mathbb{R}^2$ in which there are $K$ facilities with fixed locations $x_k \in\mathcal{X}$. The random variable $X$ represents the location of a randomly selected customer, with $f$ being the population density. The random variable $Y$ takes values in $\left\{1,...,K\right\}$ and represents a choice of facility. The $k$th facility handles a proportion $p_k$ of the total population.

The classical formulation of this problem chooses $\pi$ to minimize the expected cost $\mathbb{E}_\pi\left(c\left(X,Y\right)\right)$, where typically $c\left(x,k\right) = \|x-x_k\|_2$. Other reasonably common cost variants include squared Euclidean distance \citep{Oh07} and $1$-norms \citep{MaHoSoLe18}. It is well-known that the optimal assignment rule, given a random location $X$, sets $Y\left(X\right) = \arg\min_j c\left(X,j\right) - g^*_j$ for a certain specific choice of $g^*\in\mathbb{R}^K$, where the argmin is a.s. unique under mild regularity conditions. Thus, the $k$th facility serves all customers in the set $A_k = \left\{x:Y\left(x\right) = k\right\}$, and the sets $A_1,...,A_K$ comprise a partition of $\mathcal{X}$. This partitioning structure is known in the literature as an ``additively weighted Voronoi diagram''. The weights $g^*$ can be easily computed (for virtually any cost function) using stochastic approximation \citep{GeCuPeBa16}.

In the following, we consider an illustrative example of the partitioning problem with $\mathcal{X} = \left[0,1\right]^2$ and $f$ being the uniform density on $\mathcal{X}$. There are six facilities, and the inputs $p$ and $x_k$ are pre-generated randomly. In Figure \ref{fig:example1}, we take $c\left(x,k\right) = \|x-x_k\|_2$ and use our approach to solve $\inf_\pi \mathbb{Q}^{\alpha}_\pi\left(c\left(X,Y\right)\right)$ for various choices of $\alpha$. For reference, we also include the additively weighted Voronoi diagram (i.e., the partition that minimizes expected cost) in Figure \ref{fig:ex1voronoi}. The Voronoi diagram shows six disjoint regions whose areas correspond to the inputs $p_k$. Starting with the upper-leftmost facility and moving clockwise, these areas are $0.0522$ (green), $0.2140$ (orange), $0.1864$ (cyan), $0.1091$ (yellow), $0.2296$ (purple), and $0.2087$ (blue).

\begin{figure}[t]
        \centering
        \subfigure[Voronoi partition.]{
            \includegraphics[width=0.47\textwidth]{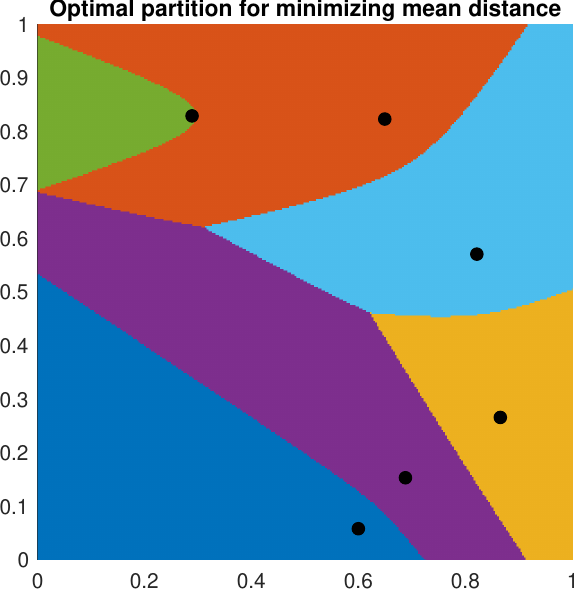}\label{fig:ex1voronoi}
        }
        \subfigure[$\alpha = 0.5$ (median).]{
            \includegraphics[width=0.47\textwidth]{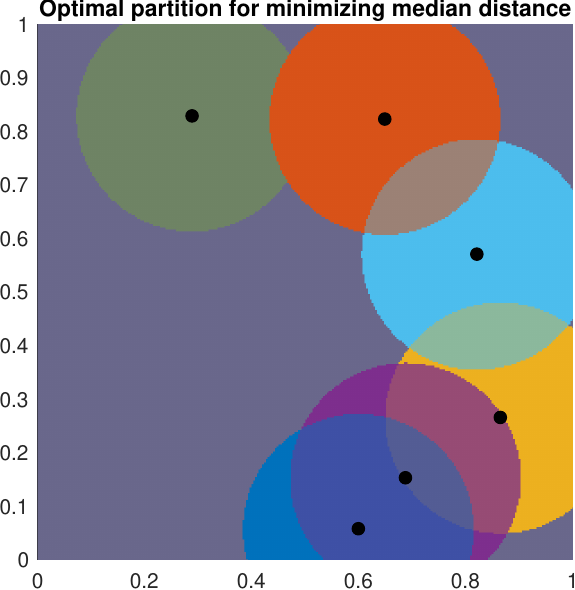}\label{fig:ex1median}
        }
        \subfigure[$\alpha = 0.25$.]{
            \includegraphics[width=0.47\textwidth]{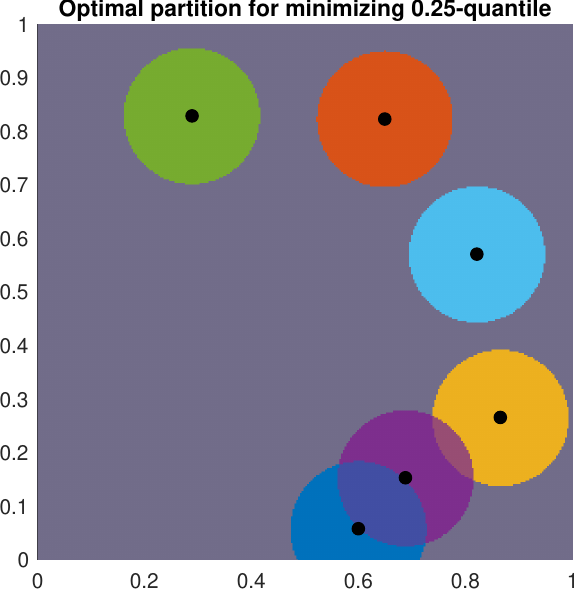}\label{fig:ex1low}
        }
        \subfigure[$\alpha = 0.95$.]{
            \includegraphics[width=0.47\textwidth]{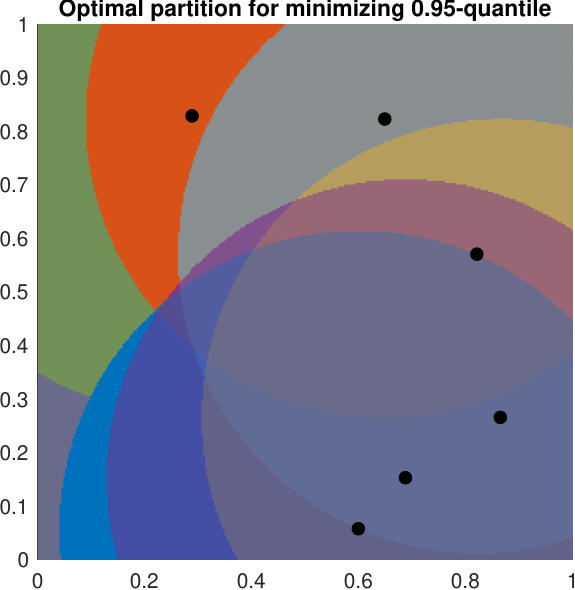}\label{fig:ex1high}
        }
        \caption{Illustrative example of optimal assignments for $c\left(x,k\right) = \|x-x_k\|_2$.}\label{fig:example1}
\end{figure}

Figure \ref{fig:ex1median} shows the optimal division of the region when the objective is to minimize the median travel distance (that is, $\alpha = 0.5$). At a high level, the geometric structure of the optimal assignment is the same for all values of $\alpha$, but let us focus on the median to make the discussion more concrete. Using a sample size $N = 10^4$ for SAA, and the CVX package \citep{cvx} with the Mosek solver, we can compute an estimate $t^* \approx 0.2099$. For each facility $k$, we draw a circle $E_k = \left\{x:\|x-x_k\|_2 \leq t^*\right\}$ of radius $t^*$ centered at $x_k$. The optimal $\psi^*_k$ in this example are equal to $1$ for all $k$ except the upper-leftmost (green) facility, which has $\psi^*_k = 0$. The main reason for this difference is that $p_k$ is very small for the green facility, requiring us to select it less frequently than the others.

For every $x$, let $B\left(x\right) = \left\{k:x\in E_k\right\}$. We assign $x$ to a facility as follows:
\begin{itemize}
\item If $B\left(x\right) = \emptyset$, we assign $x$ by randomizing over those $k$ with $\psi^*_k = 1$. In Figure \ref{fig:ex1median}, these $x$ are the points that do not lie in any of the circles. The color of this portion of the map is a ``weighted average'' of the colors used for the orange, cyan, yellow, purple, and blue regions (but not the green one), with the weights proportional to the optimal $\theta^*_k$ from Section \ref{sec:comptie}.
\item If $\left|B\left(x\right)\right| \geq 1$ and $\psi^*_k = 0$ for all $k \in B\left(x\right)$, we assign $x$ by randomizing over $B\left(x\right) \cup\left\{j:\psi^*_j=1\right\}$. In Figure \ref{fig:ex1median}, this is the portion of the green circle that does not overlap with the orange one. The color is not pure green because it is a weighted average of all six colors.
\item If $\left|B\left(x\right)\right| \geq 1$ and at least one $k\in B\left(x\right)$ has $\psi^*_k = 1$, we randomize only over $\left\{k\in B\left(x\right):\psi^*_k=1\right\}$. Thus, in Figure \ref{fig:ex1median}, the overlap between, e.g., the orange and cyan circles only has those two colors, while the overlap between the orange and green circles is entirely orange. The pure orange, cyan, yellow, purple and blue portions of those respective circles represent locations where the assignment is deterministic.
\end{itemize}
Note, however, that we still need to run the algorithms in Section \ref{sec:comp} even if we know this general structure, because that is the only way to obtain the correct $t^*$, $\psi^*$, and $\theta^*$, or a close approximation thereof.

We see the same structure in Figures \ref{fig:ex1low}-\ref{fig:ex1high}, where $\alpha$ is set to $0.25$ and $0.95$, respectively. Although one typically prefers to optimize upper quantiles in cost minimization, we include $\alpha = 0.25$ purely for visual insight. The main difference between these two cases and $\alpha = 0.5$ is that the radius $t^*$ of each circle (i.e., the optimal quantile) is smaller in the first case (relative to $\alpha = 0.5$) and larger in the second case. For $\alpha = 0.25$, we have $\psi^*_k = 1$ for all $k$ and so the green circle uses only one color. For $\alpha = 0.95$, we again have $\psi^*_k = 1$ for all facilities except the green one, which has $\psi^*_k = 0$. A consequence is that the locations close to the orange facility are assigned randomly, because they exist in the overlap between the orange and cyan circles, but the locations close to the green facility are deterministically assigned to the \textit{orange} one, because they exist only in the overlap between orange and green.

\begin{figure}[t]
        \centering
        \subfigure[$\alpha = 0.5$.]{
            \includegraphics[width=0.47\textwidth]{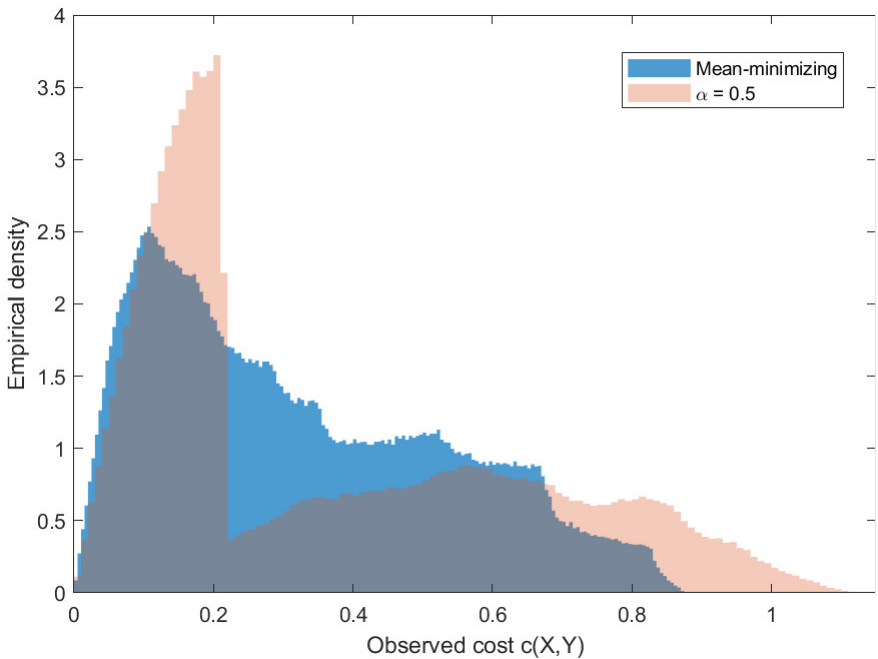}\label{fig:dist1}
        }
        \subfigure[$\alpha = 0.95$.]{
            \includegraphics[width=0.47\textwidth]{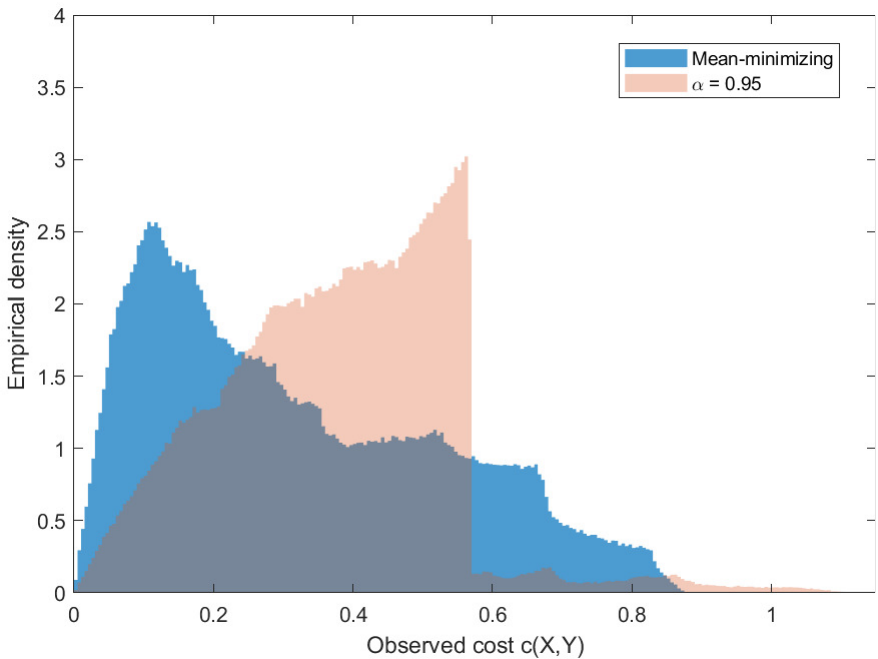}\label{fig:dist2}
        }
        \caption{Empirical densities of $c\left(X,Y\right)$ for different values of $\alpha$. The mean-minimizing partition is included in each plot for easy visual reference.}\label{fig:dist}
\end{figure}

\begin{table}[b]
\centering
\begin{tabular}{|c||c|c|c|}
\hline
Method & Mean cost & Median cost & $95$th percentile\tabularnewline
\hline
\hline
Mean-minimizing & 0.3221 & 0.2753 & 0.7084\tabularnewline
\hline
$\alpha = 0.5$ & 0.3666 & 0.2148 & 0.8781\tabularnewline
\hline
$\alpha = 0.95$ & 0.3777 & 0.3864 & 0.5667\tabularnewline
\hline
\end{tabular}
\caption{\label{tab:stats}Statistics for $c\left(X,Y\right)$ under three distributions of $Y$.}
\end{table}

Figure \ref{fig:dist} shows the empirical densities of $c\left(X,Y\right)$ when $Y$ is chosen according to the solution of quantile optimization. We compare $\alpha = 0.5$ (median-minimizing), $\alpha = 0.95$, and also include the mean-minimizing distribution for reference. We omit $\alpha = 0.25$ from the comparison for brevity. The mean, median, and $95$th percentile of $c\left(X,Y\right)$ under each distribution are summarized in Table \ref{tab:stats}. In Figure \ref{fig:dist1}, we see that the median-minimizing solution prefers to reduce larger costs without much regard for how far out they are in the tails; thus, there are quite a few extreme values, but their incidence is low. On the other hand, in Figure \ref{fig:dist2}, the incidence of very high costs is very sharply reduced, at the expense of the median.

It is also interesting to see how the visualization of the optimal assignment changes depending on the cost function. Figure \ref{fig:example2} presents a second example with different $p$ and $x_k$ as well as a different $c\left(x,k\right) = \|x-x_k\|_1$. The $1$-norm changes the shape of the regions in the Voronoi diagram (Figure \ref{fig:ex2voronoi}). However, when the objective is to optimize the median of the $1$-norm, the high-level structure is the same as in the previous example, with the only difference being that the circles are now diamonds. Once more, the lower-leftmost facility is the only one with $\psi^*_k = 0$, because its low $p_k$ value requires us to select it less frequently. We do not show results for other quantiles here, since they generally repeat the insights from the previous example, but it is certainly possible to optimize them if desired.

\begin{figure}[t]
        \centering
        \subfigure[Voronoi partition.]{
            \includegraphics[width=0.47\textwidth]{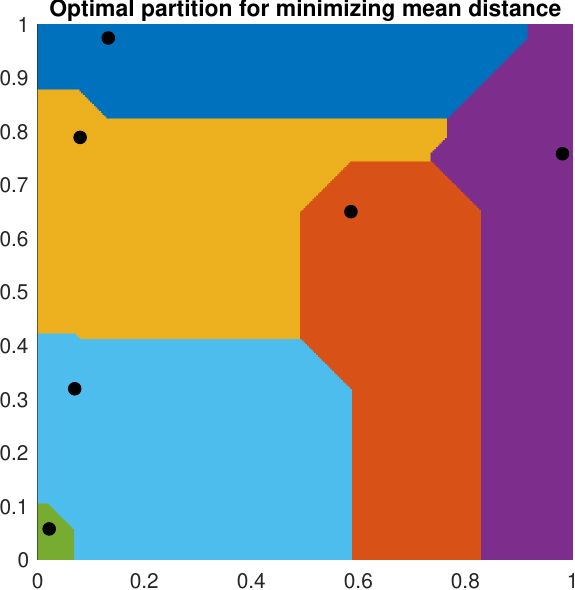}\label{fig:ex2voronoi}
        }
        \subfigure[$\alpha = 0.5$ (median).]{
            \includegraphics[width=0.47\textwidth]{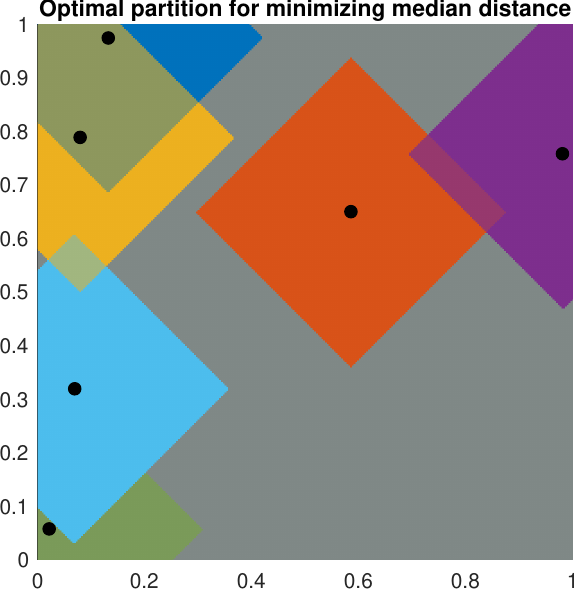}\label{fig:ex2median}
        }
        \caption{Illustrative example of optimal assignments for $c\left(x,k\right) = \|x-x_k\|_1$.}\label{fig:example2}
\end{figure}

To our knowledge, the geometric structure in Figures \ref{fig:example1}-\ref{fig:example2} has never before been observed in the literature on either optimal transport or computational geometry. These examples show that, in the setting of quantile optimization, not only can the optimal assignments be characterized theoretically, they are also computationally tractable using our methods.

\section{Conclusion}\label{sec:conc}

To our knowledge, this is the first paper to consider quantile optimization in any form of optimal transport. The semidiscrete setting is a natural first choice to develop these ideas because it enables cleaner theoretical characterization and efficient computation. While some of the structure in our paper may be generalizable, it is highly unlikely that optimal plans would be computable under general continuous distributions. Even under the classical objective of minimizing expected cost, it is virtually impossible to obtain computationally tractable plans without resorting to discretization. Therefore, from a practical point of view, the semidiscrete setting may in fact be the most relevant to study.

Even in the semidiscrete case, quantile optimization has shown itself to be more complex than mean minimization because the optimal plan does not have the Monge property. In other words, given $X$, the optimal choice of $Y$ is not deterministic, but may require randomization via a tie-breaking rule. Such rules are quite under-studied in the optimal transport literature, where one usually works with assumptions under which Kantorovich-Monge equivalence holds. We believe that our approach to characterizing and computing the tie-breaking rule is one of the most novel aspects of our work. This aspect may be useful outside the quantile optimization setting, in any situation where unusual or non-standard cost functions preclude the existence of an optimal deterministic map. We hope that our work will lead to increased interest in such scenarios.

\bibliographystyle{agsm} 
\bibliography{quantile} 

\newpage

\section{Appendix: proofs}

Below, we give complete proofs for all results that were stated in the main text.

\subsection{Proof of Theorem \ref{thm:farkas}}

First, we show by contradiction that both alternatives cannot be true simultaneously. Suppose that $\pi$ satisfies (\ref{eq:equality}) and (\ref{eq:target})-(\ref{eq:nonneg}), while $\left(\zeta,\phi,\psi\right)$ satisfies (\ref{eq:dualcons})-(\ref{eq:farkas}). Since $\pi\left(x,k\right) \geq 0$, (\ref{eq:dualcons}) yields
\begin{equation*}
1_{\left\{c\left(x,k\right)\leq t\right\}}\pi\left(x,k\right)\zeta + \phi\left(x\right)\pi\left(x,k\right) + \psi_k\pi\left(x,k\right) \geq 0.
\end{equation*}
Summing over $k$ and integrating over $x$ yields
\begin{equation*}
\zeta\int_{x\in\mathcal{X}} \sum_k 1_{\left\{c\left(x,k\right)\leq t\right\}}\pi\left(x,k\right)dx + \int_{x\in\mathcal{X}}\left(\phi\left(x\right)\sum_k \pi\left(x,k\right)\right)dx + \sum_k \psi_k \int_{x\in\mathcal{X}} \pi\left(x,k\right)dx \geq 0.
\end{equation*}
By (\ref{eq:target})-(\ref{eq:density}) and (\ref{eq:equality}), this expression simplifies to
\begin{equation*}
\zeta\alpha + \int_{x\in\mathcal{X}}\phi\left(x\right)f\left(x\right)dx + \sum_k p_k \psi_k \geq 0,
\end{equation*}
which contradicts (\ref{eq:farkas}).

Now, to complete the proof, it is sufficient to suppose that the system described by (\ref{eq:equality}) and (\ref{eq:target})-(\ref{eq:nonneg}) has no solution, and then to show that the system described by (\ref{eq:dualcons})-(\ref{eq:farkas}) must have a solution. Define the cone
\begin{equation*}
M\hspace{-0.05in} =\hspace{-0.05in} \left\{\left(\sum_k w\left(\cdot,k\right),\int_{x\in\mathcal{X}}w\left(x,\cdot\right)dx,\int_{x\in\mathcal{X}}\sum_k 1_{\left\{c\left(x,k\right)\leq t\right\}}w\left(x,k\right)dx\right):w\left(\cdot,\cdot\right)\geq 0,\int_{x\in\mathcal{X}}\sum_k w\left(x,k\right)^2 dx < \infty\right\}.
\end{equation*}
Let $m_* = \left(f\left(\cdot\right),p,\alpha\right)$. Consider the Hilbert space
\begin{equation*}
H = \left\{\left(h\left(\cdot\right),v,a\right):\int_{x\mathcal{X}}h\left(x\right)^2 dx< \infty,\,v\in\mathbb{R}^K,\,a\in\mathbb{R}\right\}
\end{equation*}
with the inner product
\begin{equation*}
\langle\left(h_1,v_1,a_1\right),\left(h_2,v_2,a_2\right)\rangle = \int_{x\in\mathcal{X}} h_1\left(x\right)h_2\left(x\right)dx + v^\top_1 v_2 + a_1a_2.
\end{equation*}
Observe that $\left\{m_*\right\}$ and $M$ are two subsets of $H$ that are both closed and convex; furthermore, $\left\{m_*\right\}$ is also compact. Because the system described by (\ref{eq:equality}) and (\ref{eq:target})-(\ref{eq:nonneg}) has no solution, we have $m_*\notin M$, which means that $\left\{m_*\right\}$ and $M$ are disjoint. By a geometric version of the Hahn-Banach theorem (Prop. IV.3.6 and Thm. IV.3.9 in \citealp{Co19}), there exists a continuous linear functional $r:H\rightarrow\mathbb{R}$ and $\delta \in\mathbb{R}$ such that $r\left(m\right)\geq \delta$ and $r\left(m_*\right) < \delta$ for any $m\in M$.

Since we can take $w\equiv 0$, we have $\left(0,0,0\right) \in M$ and thus $\delta \leq 0$. From this it follows that $r\left(m_*\right) < 0$, and also $\inf_{m\in M} r\left(m\right)\geq 0$. To show the second property, suppose that $r\left(m\right) < 0$ for some $m \in M$. Then, since $bm\in M$ for any $b > 0$, it follows that $\delta = -\infty$, which contradicts the fact that $\delta \in\mathbb{R}$.

By Theorem 5.25 of \cite{Fo99}, every continuous linear functional on the Hilbert space is identified with an element of the Hilbert space. Therefore, there exists $\left(\phi,\psi,\zeta\right)\in H$ such that
\begin{equation*}
r\left(h,v,a\right) = \int_{x\in\mathcal{X}} \phi\left(x\right)h\left(x\right)dx + \psi^\top v + a\zeta.
\end{equation*}
Since $\inf_{m\in M}r\left(m\right) \geq 0$, we may take any $\pi\left(\cdot,\cdot\right)\geq 0$ with $\int_{x\in\mathcal{X}}\sum_k \pi\left(x,k\right)^2 dx < \infty$ and obtain
\begin{equation*}
\int_{x\in\mathcal{X}} \phi\left(x\right)\sum_k \pi\left(x,k\right)dx + \sum_k \psi_k\int_{x\in\mathcal{X}}\pi\left(x,k\right)dx + \zeta\int_{x\in\mathcal{X}}\sum_k 1_{\left\{c\left(x,k\right)\leq t\right\}}\pi\left(x,k\right)dx \geq 0,
\end{equation*}
or, equivalently,
\begin{equation*}
\mathbb{E}\left(\phi\left(X\right)\right) + \mathbb{E}\left(\psi_Y\right) + \zeta \mathbb{E}\left(1_{\left\{c\left(X,Y\right)\leq t\right\}}\right) \geq 0
\end{equation*}
for any joint distribution of $\left(X,Y\right)$. Consequently, (\ref{eq:dualcons}) must hold for all $x$ and $k$. On the other hand, since $r\left(m_*\right)<0$, we have
\begin{equation*}
\int_{x\in\mathcal{X}} g\left(x\right)f\left(x\right)dx + \sum_k p_k u_k + \alpha b < 0,
\end{equation*}
which establishes (\ref{eq:farkas}).

\subsection{Proof of Proposition \ref{prop:binary}}

We first show that, for all $k$, we may suppose that $\tilde{\psi}_k \in \left[0,1\right]$ without loss of generality. Let $\tilde{\psi}$ be a minimizer and suppose that, for some $j$, $\tilde{\psi}_j < \max_k \tilde{\psi}_k - 1$. Now, let $\psi'$ satisfy $\psi'_k = \tilde{\psi}_k$ for $k \neq j$, and $\psi'_j = \max_k \tilde{\psi}_k - 1$. Then,
\begin{equation*}
1_{\left\{c\left(X,j\right)\leq t\right\}} + \tilde{\psi}_j \leq 1 + \psi'_j = \max_k \tilde{\psi_k} \leq \max_k 1_{\left\{c\left(X,k\right)\leq t\right\}} + \tilde{\psi_k},
\end{equation*}
whence
\begin{equation*}
\mathbb{E}\left(\max_k 1_{\left\{c\left(X,k\right)\leq t\right\}} + \tilde{\psi}_k\right) = \mathbb{E}\left(\max_k 1_{\left\{c\left(X,k\right)\leq t\right\}} + \psi'_k\right).
\end{equation*}
However, $p^\top \psi' > p^\top \tilde{\psi}$, so the objective value of $\psi'$ is strictly better than that of $\tilde{\psi}$, contradicting the fact that $\tilde{\psi}$ is a minimizer. Therefore, we must have $\tilde{\psi}_j \geq \max_k \tilde{\psi}_k - 1$ for all $j$. Since the objective function is invariant to translation, i.e., adding the same constant to each element of $\psi'$ in (\ref{eq:alphastar}) does not change the objective value, we may suppose that $\tilde{\psi}_k \in \left[0,1\right]$ for all $k$.

Now, let $A_k = 1_{\left\{c\left(X,k\right)\leq t\right\}}$ and observe that the random vector $A = \left(A_1,...,A_K\right)$ has discrete support contained in $\left\{0,1\right\}^K$. We may therefore rewrite the right-hand side of (\ref{eq:alphastar}) as
\begin{equation*}
\min_{\psi} \sum_{a\in\left\{0,1\right\}^K} P\left(A = a\right)\cdot \max_k \left(a_k + \psi_k\right) - p^\top \psi.
\end{equation*}
This problem can be rewritten as the linear program $\min_\psi \sum_a P\left(A=a\right) \eta_a - p^\top \psi$ subject to $\eta_a \geq a_k + \psi_k$ for all $a,k$. Note that, for fixed $a,k$, the values $a_k$ are also fixed. Therefore, we may rewrite the constraints as
\begin{equation}\label{eq:tu}
\eta_a - \psi_k \geq a_k,
\end{equation}
where the right-hand side values are binary. If we write the left-hand side in matrix form as $B\beta$, where $\beta$ is the concatenation of $\eta$ and $\psi$, we see that $B^\top$ is the incidence matrix of a bipartite graph where, for each $\left(a,k\right)$ pair, an edge connects a node indexed by $k$ to a node indexed by $a$. Therefore, $B$ is totally unimodular. By the discussion on p. 268 of \cite{Sc98}, the system (\ref{eq:tu}) combined with $0\leq\psi_k \leq 1$ and $0\leq \eta_a \leq 2$ represents a polyhedron with all integer-valued vertices. Since any LP always has an optimal extreme-point solution, there exists a minimizer $\tilde{\psi}$ with $\tilde{\psi}_k\in\left\{0,1\right\}$ for all $k$.

\subsection{Proof of Proposition \ref{prop:continuity}}

By Proposition \ref{prop:binary}, the variable $\psi'$ on the right-hand side of (\ref{eq:alphastar}) may be restricted to the finite set $\left\{0,1\right\}^K$ without loss of generality. For any $a \in \left\{0,1\right\}^K$, let $A = \left\{k:a_k = 1\right\}$. Then,
\begin{equation*}
\max_k 1_{\left\{c\left(X,k\right)\leq t\right\}} + a_k = \left\{
\begin{array}{c l}
1_{\left\{\min_k c\left(X,k\right)\leq t\right\}} & A = \emptyset,\\
1 + 1_{\left\{\min_{k\in A} c\left(X,k\right) \leq t\right\}} & A \neq \emptyset.
\end{array}
\right.
\end{equation*}
Taking the minimum over all $a \in\left\{0,1\right\}^K$, we rewrite (\ref{eq:alphastar}) as
\begin{equation}\label{eq:finitemin}
\alpha^*\left(t\right) = \min\left\{P\left(\min_k c\left(X,k\right)\leq t\right), \min_{A\subseteq\left\{1,...,K\right\},A\neq\emptyset} 1 + P\left(\min_{k\in A} c\left(X,k\right)\leq t\right) - \sum_{k\in A}p_k\right\}.
\end{equation}
Observe that, for any nonempty $A\subseteq\left\{1,...,K\right\}$, the cdf $P\left(\min_{k\in A}c\left(X,k\right) \leq t\right)$ is continuous in $t$ because all $c\left(X,k\right)$ are atomless and $A$ is a finite set. Thus, (\ref{eq:finitemin}) is a pointwise minimum of a finite number of differentiable functions, and therefore continuous. The boundedness of the densities provides the Lipschitz property.

\subsection{Proof of Proposition \ref{prop:monotonicity}}

From the assumptions, it is easy to see that, for any nonempty $A\subseteq\left\{1,...,K\right\}$, we have
\begin{equation}\label{eq:setineq}
P\left(c\left(X,k\right) > t_2,\,\forall k\in A\right) < P\left(c\left(X,k\right) > t_1,\,\forall k\in A\right),
\end{equation}
since the marginal density of $\left\{c\left(X,k\right):k\in A\right\}$ will remain bounded away from zero in a neighborhood of the point $\left(t_1,...,t_1\right)\in\mathbb{R}^{\left|A\right|}$. Therefore, it remains to derive the desired conclusion from (\ref{eq:setineq}).

As an intermediate step, we fix $\psi\in\left\{0,1\right\}^K$ and show
\begin{equation}\label{eq:expineq}
\mathbb{E}\left(\max_k 1_{\left\{c\left(X,k\right)\leq t_1\right\}} + \psi_k\right) < \mathbb{E}\left(\max_k 1_{\left\{c\left(X,k\right)\leq t_2\right\}} + \psi_k\right),
\end{equation}
recalling that it is sufficient to consider binary-valued $\psi$ by Proposition \ref{prop:binary}. Denote $F\left(t\right) = \max_k 1_{\left\{c\left(X,k\right)\leq t\right\}} + \psi_k$. Let $A = \left\{k:\psi_k=1\right\}$. We consider two cases: $A\neq\emptyset$ and $A = \emptyset$. In the first case, we write
\begin{equation*}
F\left(t\right) = \max\left\{\max_{k\in A^c} 1_{\left\{c\left(X,k\right)\leq t\right\}},1 + \max_{k\in A}1_{\left\{c\left(X,k\right)\leq t\right\}}\right\}.
\end{equation*}
Define $M\left(t\right) = \max_{k\in A}1_{\left\{c\left(X,k\right)\leq t\right\}}$. Then, $F\left(t\right) = 1$ if and only if $M\left(t\right) = 0$, and, moreover, $F\left(t\right) = 2$ if and only if $M\left(t\right) = 1$. Therefore,
\begin{eqnarray*}
\mathbb{E}\left(F\left(t\right)\right) &=& P\left(F\left(t\right) = 1\right) + 2\cdot P\left(F\left(t\right) = 2\right)\\
&=& P\left(M\left(t\right) = 0\right) + 2\cdot P\left(M\left(t\right)=1\right)\\
&=& P\left(M\left(t\right) = 0\right) + 2\cdot \left(1 - P\left(M\left(t\right)=0\right)\right)\\
&=& 2 - P\left(M\left(t\right) = 0\right).
\end{eqnarray*}
Consequently,
\begin{eqnarray}
\mathbb{E}\left(F\left(t_2\right)\right) - \mathbb{E}\left(F\left(t_1\right)\right) &=& P\left(\max_{k\in A}1_{\left\{c\left(X,k\right)\leq t_1\right\}} = 0\right) - P\left(\max_{k\in A}1_{\left\{c\left(X,k\right)\leq t_2\right\}} = 0\right),\nonumber\\
&=& P\left(c\left(X,k\right) > t_1,\,\forall k \in A\right) - P\left(c\left(X,k\right) > t_2,\,\forall k \in A\right)\label{eq:rewritediff}
\end{eqnarray}
and, by (\ref{eq:setineq}), the difference in (\ref{eq:rewritediff}) is strictly positive.

Now, we consider the case where $A=\emptyset$. Then, $\psi \equiv 0$ and $F\left(t\right) = \max_k 1_{\left\{c\left(X,k\right)\leq t\right\}}$. It follows that
\begin{equation*}
\mathbb{E}\left(F\left(t\right)\right) = P\left(\max_k 1_{\left\{c\left(X,k\right)\leq t\right\}} = 1\right) = 1 - P\left(c\left(X,k\right) > t,\,\forall k\right),
\end{equation*}
whence
\begin{equation*}
\mathbb{E}\left(F\left(t_2\right)\right) - \mathbb{E}\left(F\left(t_1\right)\right) = P\left(c\left(X,k\right) > t_1\,\forall k\right) - P\left(c\left(X,k\right) > t_2\,\forall k\right),
\end{equation*}
which again is strictly positive by (\ref{eq:setineq}). Thus, (\ref{eq:expineq}) is proved.

Finally, since $\left\{0,1\right\}^K$ is a finite set, there exists $\delta > 0$ such that
\begin{equation}\label{eq:unifineq}
\mathbb{E}\left(\max_k 1_{\left\{c\left(X,k\right)\leq t_1\right\}} + \psi_k\right) + \delta \leq \mathbb{E}\left(\max_k 1_{\left\{c\left(X,k\right)\leq t_2\right\}} + \psi_k\right), \quad \forall \psi\in\left\{0,1\right\}^K.
\end{equation}
Take
\begin{equation*}
\psi'\in\arg\min_{\psi\in\left\{0,1\right\}^K} \mathbb{E}\left(\max_k 1_{\left\{c\left(X,k\right)\leq t_2\right\}} + \psi_k\right) - p^\top \psi.
\end{equation*}
Then, applying (\ref{eq:unifineq}) to $\psi'$, we obtain
\begin{eqnarray*}
\min_{\psi\in\left\{0,1\right\}^K} \mathbb{E}\left(\max_k 1_{\left\{c\left(X,k\right)\leq t_1\right\}} + \psi_k\right) - p^\top \psi +\delta &\leq& \mathbb{E}\left(\max_k 1_{\left\{c\left(X,k\right)\leq t_1\right\}} + \psi'_k\right) - p^\top \psi' + \delta\\
&\leq& \mathbb{E}\left(\max_k 1_{\left\{c\left(X,k\right)\leq t_2\right\}} + \psi'_k\right) - p^\top \psi',
\end{eqnarray*}
which is precisely $\alpha^*\left(t_1\right) + \delta \leq \alpha^*\left(t_2\right)$, as required.

\subsection{Proof of Theorem \ref{thm:conc}}

Define
\begin{equation*}
G\left(\psi,t\right) = \max_k \left(1_{\left\{c\left(X,k\right)\leq t\right\}} + \psi_k\right) - p^\top \psi
\end{equation*}
and $g\left(\psi,t\right) = \mathbb{E}\left(G\left(\psi,t\right)\right)$. By definition, $\min_{\psi\in\left\{0,1\right\}^K} g\left(\psi,t\right) = \alpha^*\left(t\right)$, noting that it is sufficient to consider binary-valued $\psi$ by Proposition \ref{prop:binary}. Let
\begin{equation*}
\Delta^N\left(\psi,t\right) = \sqrt{N}\left(\mathbb{E}^N \left(G\left(\psi,t\right)\right) - \mathbb{E}\left(G\left(\psi,t\right)\right)\right),
\end{equation*}
where $\mathbb{E}^N$ denotes the empirical expectation, i.e., the expectation over the discrete distribution where $X = X^n$ with probability $\frac{1}{N}$ for $n = 1,...,N$. We may write
\begin{equation*}
\hat{\alpha}^N\left(t\right) = \min_{\psi\in\left\{0,1\right\}^K} g\left(\psi,t\right) + N^{-\frac{1}{2}}\Delta^N\left(\psi,t\right).
\end{equation*}
By the elementary inequality $\left|\min \left(h_1+h_2\right) - \max h_1\right| \leq \max \left|h_2\right|$, we have
\begin{equation*}
\sqrt{N}\left|\hat{\alpha}^N\left(t\right) - \alpha^*\left(t\right)\right| \leq \sup_{\psi\in\left\{0,1\right\}^K} \left|\Delta^N\left(\psi,t\right)\right|
\end{equation*}
for all $t$. Consequently,
\begin{equation*}
\sqrt{N}\sup_{t\in\mathbb{R}}\left|\hat{\alpha}^N\left(t\right) - \alpha^*\left(t\right)\right| \leq \sup_{t\in\mathbb{R}}\sup_{\psi\in\left\{0,1\right\}^K} \left|\Delta^N\left(\psi,t\right)\right|.
\end{equation*}
Using techniques from empirical process theory, the following technical lemma can be obtained. The proof is given in a separate subsection of the Appendix.

\begin{lem}\label{lem:technical}
There exist constants $D_1,D_2 > 0$ such that
\begin{equation}\label{eq:techlemma}
P\left(\sup_{\psi\in\left\{0,1\right\}^K} \sup_t \left|\Delta^N\left(\psi,t\right)\right| > w\right) \leq D_1 e^{-D_2 w^2}
\end{equation}
for all $w > 0$.
\end{lem}

From Lemma \ref{lem:technical}, we obtain the exponential bound
\begin{equation}\label{eq:expbound}
P\left(\|\hat{\alpha}^N - \alpha^*\|_{\infty} > \frac{w}{\sqrt{N}}\right) \leq D_1 e^{-D_2 w^2}.
\end{equation}
Then, we observe that
\begin{equation}\label{eq:hatineqs}
\alpha - \|\hat{\alpha}^N - \alpha^*\|_{\infty} \leq \hat{\alpha}^N\left(\hat{t}^N\right) - \|\hat{\alpha}^N - \alpha^*\|_{\infty} \leq \alpha^*\left(\hat{t}^N\right),
\end{equation}
where the first inequality follows by (\ref{eq:hatroot}).

The random variables $c\left(X^m,j\right)$ and $c\left(X^n,k\right)$ are independent as long as $m\neq n$. Therefore, because $c\left(X,k\right)$ is a continuous random variable for each $k$, we may assume without loss of generality that $c\left(X^m,j\right) \neq c\left(X^n,k\right)$. For fixed $\psi$ and $n$, the mapping $t\mapsto G\left(\psi,t\right)$ has jumps of size at most $1$, whence $t\mapsto \mathbb{E}^N G\left(\psi,t\right)$ has jumps of size at most $\frac{1}{N}$. Since $\hat{\alpha}^N\left(t\right) = \min_{\psi\in\left\{0,1\right\}^K} \mathbb{E}^N G\left(\psi,t\right)$ is the pointwise minimum of finitely many such mappings, its jumps are also bounded by $\frac{1}{N}$.

By definition of $\hat{t}^N$, we have $\hat{\alpha}^N\left(\hat{t}^N\right) \geq \alpha$, and by definition of the infimum, we have $\hat{\alpha}^N\left(\lim_{t\nearrow \hat{t}^N} t\right) \leq \alpha$. Consequently,
\begin{equation*}
\hat{\alpha}^N\left(\hat{t}^N\right) - \alpha < \frac{1}{N}.
\end{equation*}
By $\hat{\alpha}^N\left(\hat{t}^N\right) \geq \alpha^*\left(\hat{t}^N\right) - \|\hat{\alpha}^N - \alpha^*\|_{\infty}$, we also have
\begin{equation*}
\alpha^*\left(\hat{t}^N\right) \leq \alpha + \|\hat{\alpha}^N - \alpha^*\|_{\infty} + \frac{1}{N}.
\end{equation*}
Combining this with (\ref{eq:hatineqs}) yields
\begin{equation*}
\left|\alpha^*\left(\hat{t}^N\right) - \alpha\right| \leq \|\hat{\alpha}^N - \alpha^*\|_{\infty} + \frac{1}{N}.
\end{equation*}
Applying (\ref{eq:expbound}) and choosing suitably large $C_1,C_2 > 0$ completes the proof.

\subsection{Proof of Theorem \ref{thm:conct}}

From Proposition \ref{prop:monotonicity}, we know that there exist $t_1 < t^* < t_2$ such that $\alpha^*\left(t^*\right)$ is strictly increasing on $\left[t_1,t_2\right]$. We may take a small $\delta > 0$ such that $t_1 < t^* - \delta < t^* + \delta < t_2$. Define
\begin{equation}\label{eq:bdelta}
b\left(\delta\right) = \min\left\{\alpha^*\left(t^*+\delta\right)-\alpha,\alpha - \alpha^*\left(t^*-\delta\right)\right\}.
\end{equation}
We claim that, if $\|\hat{\alpha}^N - \alpha^*\|_{\infty} \leq \frac{1}{2}b\left(\delta\right)$, then $\left|\hat{t}^N-t^*\right| \leq \delta$.

To prove the claim, suppose that $\|\hat{\alpha}^N - \alpha^*\|_{\infty} \leq \frac{1}{2}b\left(\delta\right)$ and take $t < t^* - \delta$. Then, $\alpha^*\left(t\right) < \alpha^*\left(t^*-\delta\right)$. By (\ref{eq:bdelta}), we have
\begin{equation*}
b\left(\delta\right) \leq \alpha - \alpha^*\left(t^*-\delta\right) \qquad \Rightarrow \qquad \alpha^*\left(t^*-\delta\right) \leq \alpha - b\left(\delta\right).
\end{equation*}
Then,
\begin{equation*}
\hat{\alpha}^N\left(t\right) \leq \alpha^*\left(t\right) + \frac{1}{2}b\left(\delta\right) \leq \alpha - \frac{1}{2}b\left(\delta\right) < \alpha.
\end{equation*}
Therefore, for all $t < t^*-\delta$, we have $\hat{\alpha}^N\left(t\right) < \alpha$. Since $\hat{\alpha}^N\left(\hat{t}^N\right) \geq \alpha$ by definition, it follows that $\hat{t}^N \geq t^*-\delta$.

For the other direction, (\ref{eq:bdelta}) yields
\begin{equation*}
b\left(\delta\right) \leq \alpha^*\left(t^*+\delta\right)-\alpha \qquad \Rightarrow \qquad \alpha^*\left(t^*+\delta\right) \geq \alpha + b\left(\delta\right).
\end{equation*}
We then have
\begin{equation*}
\hat{\alpha}^N\left(t+\delta\right) \geq \alpha^*\left(t+\delta\right) - \frac{1}{2}b\left(\delta\right) \geq \alpha + \frac{1}{2}b\left(\delta\right) > \alpha.
\end{equation*}
So, by the definition of $\hat{t}^N$ we have $\hat{t}^N \leq t^*+\delta$. Thus, the claim is proved.

As a consequence, we have
\begin{equation*}
P\left(\left|\hat{t}^N-t\right|>\delta\right) \leq P\left(\|\hat{\alpha}^N - \alpha^*\|_{\infty} > \frac{1}{2}b\left(\delta\right)\right) \leq D_1 e^{-\frac{M_2 b\left(\delta\right)^2}{4}N},
\end{equation*}
where the last inequality (for appropriate $M_1,M_2 > 0$) is obtained by repeating the arguments in the proof of Theorem \ref{thm:conc} (essentially Lemma \ref{lem:technical}). The proof is then completed by showing that $b\left(\delta\right)\sim \mathcal{O}\left(\delta\right)$. We show this in the following technical lemma, proved in a separate subsection of the Appendix. The lemma covers only one of the terms inside the minimum in (\ref{eq:bdelta}), but the other term is handled symmetrically.

\begin{lem}\label{lem:technical2}
Suppose that we are in the situation of Theorem \ref{thm:conct}. Then, there exists some $D > 0$ such that, for all sufficiently small $\delta > 0$, we have $\alpha^*\left(t^*+\delta\right) - \alpha \geq D\delta$.
\end{lem}

\subsection{Proof of Theorem \ref{thm:decision}}

As in other proofs, let $\mathbb{E}^N$ denote the expectation over the empirical distribution, under which $P\left(X = X^n\right) = \frac{1}{N}$ for $n = 1,...,N$. Define the gap
\begin{equation*}
B\left(t\right) = \left(\min_{\psi\notin\Psi^*\left(t\right)} \mathbb{E}\left(\max_k 1_{\left\{c\left(X,k\right)\leq t\right\}} + \psi_k\right) - p^\top \psi\right) - \alpha^*\left(t\right).
\end{equation*}
By construction, $B\left(t^*\right) > 0$. By continuity arguments (see the proof of Proposition \ref{prop:continuity}), we have
\begin{equation*}
\inf_{t\in\left[t^*-\delta,t^*+\delta\right]} B\left(t\right) \geq B_0
\end{equation*}
for some $B_0 > 0$ and $\delta > 0$. Then, if
\begin{equation*}
\sup_{t\in\left[t^*-\delta,t^*+\delta\right]} \sup_{\psi\in\left\{0,1\right\}^K} \left|\mathbb{E}^N\left(\max_k 1_{\left\{c\left(X,k\right)\leq t\right\}} + \psi_k\right) - \mathbb{E}\left(\max_k 1_{\left\{c\left(X,k\right)\leq t\right\}} + \psi_k\right)\right| < \frac{1}{2}B_0,
\end{equation*}
we have $\hat{\Psi}^N\left(t\right) = \Psi^*\left(t\right)$ for all $t\in\left[t^*-\delta,t^*+\delta\right]$. If we also have $\left|\hat{t}^N-t^*\right|\leq \delta$, then $\hat{\Psi}^N\left(\hat{t}^N\right) = \Psi^*\left(\hat{t}^N\right)$. Therefore,
\begin{eqnarray*}
&\,& P\left(\hat{\Psi}^N\left(\hat{t}^N\right) \neq \Psi^*\left(\hat{t}^N\right)\right) \leq  P\left(\left|\hat{t}^N-t^*\right|>\delta\right)\\
&+& P\left(\sup_{t\in\left[t^*-\delta,t^*+\delta\right]} \sup_{\psi\in\left\{0,1\right\}^K} \left|\mathbb{E}^N\left(\max_k 1_{\left\{c\left(X,k\right)\leq t\right\}} + \psi_k\right) - \mathbb{E}\left(\max_k 1_{\left\{c\left(X,k\right)\leq t\right\}} + \psi_k\right)\right| > \frac{1}{2}B_0\right).
\end{eqnarray*}
The first term on the right-hand side is bounded using Theorem \ref{thm:conct}. The second term is bounded by repeating the arguments in the proof of Theorem \ref{thm:conc} (essentially Lemma \ref{lem:technical}). The desired result then follows.

\subsection{Proof of Theorem \ref{thm:subset}}

As in other proofs, let $\mathbb{E}^N$ denote the expectation over the empirical distribution, under which $P\left(X = X^n\right) = \frac{1}{N}$ for $n = 1,...,N$. As in the proof of Theorem \ref{thm:conc}, define
\begin{equation*}
G\left(\psi,t\right) = \max_k \left(1_{\left\{c\left(X,k\right)\leq t\right\}} + \psi_k\right) - p^\top \psi.
\end{equation*}
Let
\begin{equation*}
\bar{\Delta}^N = \sup_{\psi\in\left\{0,1\right\}^K} \sup_t \left|\mathbb{E}^N \left(G\left(\psi,t\right)\right) - \mathbb{E}\left(G\left(\psi,t\right)\right)\right|.
\end{equation*}
By repeating arguments in the proof of Theorem \ref{thm:conc} (essentially Lemma \ref{lem:technical}), we can obtain
\begin{equation}\label{eq:concGdiff}
P\left(\bar{\Delta}^N > \frac{w}{\sqrt{N}}\right) \leq M_1 e^{-M_2 w^2}
\end{equation}
for some $M_1,M_2 > 0$.

As in the proof of Theorem \ref{thm:decision}, define
\begin{equation*}
B\left(t\right) = \left(\min_{\psi\notin\Psi^*\left(t\right)} \mathbb{E}\left(\max_k 1_{\left\{c\left(X,k\right)\leq t\right\}} + \psi_k\right) - p^\top \psi\right) - \alpha^*\left(t\right),
\end{equation*}
recalling that $\alpha^*\left(t^*\right) = \alpha$. Without loss of generality, suppose that $\Psi^*\left(t^*\right) \neq \left\{0,1\right\}^K$, as the result would trivially hold in that case. Take arbitrary $\hat{\psi} \in \hat{\Psi}\left(\hat{t}^N\right)$. By repeating arguments in the proof of Theorem \ref{thm:conc}, we have $\hat{\alpha}^N\left(\hat{t}^N\right) - \alpha < \frac{1}{N}$. Observe that
\begin{equation*}
\min_{\psi\notin\Psi^*\left(t^*\right)} \mathbb{E}^N\left(G\left(\psi,\hat{t}^N\right)\right) = \hat{\alpha}^N\left(\hat{t}^N\right)
\end{equation*}
on the event $\left\{\hat{\psi}\notin\Psi^*\left(t^*\right)\right\}$. Therefore,
\begin{eqnarray*}
P\left(\hat{\psi}\notin\Psi^*\left(t^*\right)\right) &\leq& P\left(\min_{\psi\notin\Psi^*\left(t^*\right)} \mathbb{E}^N\left(G\left(\psi,\hat{t}^N\right)\right) < \alpha + \frac{1}{N}\right)\\
&\leq & P\left(-\bar{\Delta} + \min_{\psi\notin\Psi^*\left(t^*\right)} \mathbb{E}\left(G\left(\psi,\hat{t}^N\right)\right) < \alpha + \frac{1}{N}\right)\\
&\leq & P\left(-\bar{\Delta} + \min_{\psi\notin\Psi^*\left(t^*\right)} \mathbb{E}\left(G\left(\psi,\hat{t}^N\right)\right) - \min_{\psi\notin\Psi^*\left(t^*\right)} \mathbb{E}\left(G\left(\psi,t^*\right)\right) + \min_{\psi\notin\Psi^*\left(t^*\right)} \mathbb{E}\left(G\left(\psi,t^*\right)\right)< \alpha + \frac{1}{N}\right)\\
&=& P\left(-\bar{\Delta} + \min_{\psi\notin\Psi^*\left(t^*\right)} \mathbb{E}\left(G\left(\psi,\hat{t}^N\right)\right) - \min_{\psi\notin\Psi^*\left(t^*\right)} \mathbb{E}\left(G\left(\psi,t^*\right)\right) + \alpha + B\left(t^*\right)< \alpha + \frac{1}{N}\right)\\
&=& P\left(-\bar{\Delta} + \min_{\psi\notin\Psi^*\left(t^*\right)} \mathbb{E}\left(G\left(\psi,\hat{t}^N\right)\right) - \min_{\psi\notin\Psi^*\left(t^*\right)} \mathbb{E}\left(G\left(\psi,t^*\right)\right) -\frac{1}{N}< -B\left(t^*\right)\right)\\
&\leq & P\left(-\bar{\Delta} < -\frac{1}{3}B\left(t^*\right)\right)\\
&\,& + P\left(\min_{\psi\notin\Psi^*\left(t^*\right)} \mathbb{E}\left(G\left(\psi,\hat{t}^N\right)\right) - \min_{\psi\notin\Psi^*\left(t^*\right)} \mathbb{E}\left(G\left(\psi,t^*\right)\right)<-\frac{1}{3}B\left(t^*\right)\right) + 1_{\left\{-\frac{1}{N} < -\frac{1}{3}B\left(t^*\right)\right\}}.
\end{eqnarray*}
For large enough $N$, $1_{\left\{-\frac{1}{N} < -\frac{1}{3}B\left(t^*\right)\right\}} = 0$. Applying (\ref{eq:concGdiff}), we have
\begin{equation*}
P\left(-\bar{\Delta} < -\frac{1}{3}B\left(t^*\right)\right)\leq P\left(\left|\bar{\Delta}\right|>\frac{1}{3}B\left(t^*\right)\right)\leq M_1 e^{-\frac{M_2 B\left(t^*\right)^2}{9}N}.
\end{equation*}
By repeating arguments in the proof of Proposition \ref{prop:continuity}, we may show that $t\mapsto \min_{\psi\notin\Psi^*\left(t^*\right)} \mathbb{E}\left(G\left(\psi,\hat{t}^N\right)\right)$ is Lipschitz. Therefore, there exists $M_3 > 0$ such that
\begin{eqnarray*}
&\,& P\left(\min_{\psi\notin\Psi^*\left(t^*\right)} \mathbb{E}\left(G\left(\psi,\hat{t}^N\right)\right) - \min_{\psi\notin\Psi^*\left(t^*\right)} \mathbb{E}\left(G\left(\psi,t^*\right)\right)<-\frac{1}{3}B\left(t^*\right)\right)\\
&\leq & P\left(\left|\min_{\psi\notin\Psi^*\left(t^*\right)} \mathbb{E}\left(G\left(\psi,\hat{t}^N\right)\right) - \min_{\psi\notin\Psi^*\left(t^*\right)} \mathbb{E}\left(G\left(\psi,t^*\right)\right)<-\frac{1}{3}B\left(t^*\right)\right|> \frac{1}{3}B\left(t^*\right)\right)\\
&\leq & P\left(M_3 \left|\hat{t}^N-t^*\right| > \frac{1}{3}B\left(t^*\right)\right)\\
&\leq & C_3 e^{-\frac{C_4 B\left(t^*\right)^2}{9M^2_3}N},
\end{eqnarray*}
where $C_3,C_4$ are obtained from Theorem \ref{thm:conct}. Since these bounds do not depend on $\hat{\psi}$, the desired result follows.

\subsection{Proof of Proposition \ref{prop:newargmax}}

Recall that $S\left(t,\psi,X\right) = \arg\max_k 1_{\left\{c\left(X,k\right)\leq t\right\}} + \psi_k$. Let
\begin{equation*}
E = \left\{ \hat{\Psi}^N\left(\hat{t}^N\right) \subseteq \Psi^*\left(t^*\right)\right\}
\end{equation*}
Then,
\begin{eqnarray*}
P\left(S\left(\hat{t}^N,\hat{\psi},X\right) \neq S\left(t^*,\psi^*,X\right)\right) &\leq & P\left(E\cap\left\{ S\left(\hat{t}^N,\hat{\psi},X\right) \neq S\left(t^*,\psi^*,X\right)\right\}\right) + P\left(E^c\right)\\
&\leq & P\left(E\cap\left\{ S\left(\hat{t}^N,\hat{\psi},X\right) \neq S\left(t^*,\psi^*,X\right)\right\}\right) + C_7 e^{-C_8 N},
\end{eqnarray*}
where $C_7,C_8 > 0$ are obtained from Theorem \ref{thm:subset}. On the event $E$, we may assume that $\hat{\psi} = \psi^*$ without loss of generality. Thus, it suffices to bound
\begin{eqnarray}
P\left(S\left(\hat{t}^N,\hat{\psi},X\right) \neq S\left(t^*,\psi^*,X\right)\right) &=& P\left(S\left(\hat{t}^N,\hat{\psi},X\right) \setminus S\left(t^*,\psi^*,X\right)\neq \emptyset\right)\nonumber\\
&\,& + P\left(S\left(t^*,\psi^*,X\right) \setminus S\left(\hat{t}^N,\hat{\psi},X\right)\neq\emptyset\right).\label{eq:twosetminuses}
\end{eqnarray}
The two cases on the right-hand side of (\ref{eq:twosetminuses}) are handled symmetrically, so we may focus on
\begin{equation}\label{eq:setminus}
P\left(S\left(\hat{t}^N,\psi^*,X\right) \setminus S\left(t^*,\psi^*,X\right) \neq \emptyset\right) \leq \sum_j P\left(j \in S\left(\hat{t}^N,\psi^*,X\right) \setminus S\left(t^*,\psi^*,X\right)\right).
\end{equation}
Note that $S\left(t^*,\psi^*,X\right)$ is never empty. Therefore,
\begin{eqnarray}
&\,& P\left(j \in S\left(\hat{t}^N,\psi^*,X\right) \setminus S\left(t^*,\psi^*,X\right)\right)\nonumber\\
&=& \sum_i P\left(j \in S\left(\hat{t}^N,\psi^*,X\right) \setminus S\left(t^*,\psi^*,X\right),\,i\in S\left(t^*,\psi^*,X\right)\right)\nonumber\\
&\leq & \sum_i P\left(1_{\left\{c\left(X,j\right)\leq \hat{t}^N\right\}} + \psi^*_j \geq 1_{\left\{c\left(X,i\right)\leq \hat{t}^N\right\}} + \psi^*_i,\, 1_{\left\{c\left(X,i\right)\leq t^*\right\}} + \psi^*_i > 1_{\left\{c\left(X,j\right)\leq t^*\right\}} + \psi^*_j\right)\nonumber\\
&\leq & \sum_i P\left(1_{\left\{c\left(X,j\right)\leq \hat{t}^N\right\}} + \psi^*_j \geq 1_{\left\{c\left(X,i\right)\leq \hat{t}^N\right\}} + \psi^*_i,\, 1_{\left\{c\left(X,i\right)\leq t^*\right\}} + \psi^*_i \geq 1 + 1_{\left\{c\left(X,j\right)\leq t^*\right\}} + \psi^*_j\right),\label{eq:integersonly}
\end{eqnarray}
where (\ref{eq:integersonly}) follows from the fact that $1_{\left\{c\left(X,j\right)\leq t^*\right\}} + \psi^*_j$ is integer-valued. Letting $d_j = 1_{\left\{c\left(X,j\right)\leq \hat{t}^N\right\}} - 1_{\left\{c\left(X,j\right)\leq t^*\right\}}$, we may further bound (\ref{eq:integersonly}) as
\begin{eqnarray*}
&\,& \sum_i P\left(1_{\left\{c\left(X,j\right)\leq \hat{t}^N\right\}} + \psi^*_j \geq 1_{\left\{c\left(X,i\right)\leq \hat{t}^N\right\}} + \psi^*_i,\, 1_{\left\{c\left(X,i\right)\leq t^*\right\}} + \psi^*_i \geq 1 + 1_{\left\{c\left(X,j\right)\leq t^*\right\}} + \psi^*_j\right)\\
&=& \sum_i P\left(d_j + 1_{\left\{c\left(X,j\right)\leq \hat{t}^N\right\}} + \psi^*_j \geq d_i + 1_{\left\{c\left(X,i\right)\leq \hat{t}^N\right\}} + \psi^*_i,\, 1_{\left\{c\left(X,i\right)\leq t^*\right\}} + \psi^*_i \geq 1 + 1_{\left\{c\left(X,j\right)\leq t^*\right\}} + \psi^*_j\right)\\
&\leq & \sum_i P\left(d_j \geq 1 + d_i\right)\\
&\leq & \sum_i P\left(\left\{d_j \neq 0\right\}\cup\left\{d_i \neq 0\right\}\right)\\
&\leq & \sum_i P\left(d_j \neq 0\right) + P\left(d_i \neq 0\right).
\end{eqnarray*}
We then observe that
\begin{eqnarray}
P\left(d_j \neq 0\right) &=& P\left(\left\{\hat{t}^N < c\left(X,j\right) \leq t^*\right\}\cup\left\{t^* < c\left(X_j\right) \leq \hat{t}^N\right\}\right)\nonumber\\
&\leq & P\left(\left|c\left(X,j\right) - t^*\right| \leq \left|\hat{t}^N - t^*\right|\right)\nonumber\\
&\leq & \mathbb{E}\left(P\left(\left|c\left(X,j\right) - t^*\right| \leq \left|\hat{t}^N - t^*\right|\mid \hat{t}^N\right)\right).\label{eq:usecondprob}
\end{eqnarray}
For $\delta > 0$, we may write
\begin{equation*}
P\left(\left|c\left(X,j\right) - t^*\right| \leq \left|\hat{t}^N - t^*\right|\mid \hat{t}^N\right) \leq P\left(\left|c\left(X,j\right) - t^*\right| \leq \left|\hat{t}^N - t^*\right|\mid \hat{t}^N\right)1_{\left\{\left|\hat{t}^N-t^*\right|<\delta\right\}} + 1_{\left\{\left|\hat{t}^N-t^*\right| > \delta\right\}},
\end{equation*}
and, returning to (\ref{eq:usecondprob}) and taking $\delta$ to be sufficiently small, we obtain
\begin{eqnarray*}
P\left(d_j \neq 0\right) &\leq & \mathbb{E}\left(P\left(\left|c\left(X,j\right) - t^*\right| \leq \left|\hat{t}^N - t^*\right|\mid \hat{t}^N\right)1_{\left\{\left|\hat{t}^N-t^*\right|<\delta\right\}}\right) + C_3 e^{-C_4 \delta^2 N},
\end{eqnarray*}
where $C_3,C_4 > 0$ are obtained from Theorem \ref{thm:conct}. Finally, we derive
\begin{eqnarray}
&\,& \mathbb{E}\left(P\left(\left|c\left(X,j\right) - t^*\right| \leq \left|\hat{t}^N - t^*\right|\mid \hat{t}^N\right)1_{\left\{\left|\hat{t}^N-t^*\right|<\delta\right\}}\right)\nonumber\\
&\leq & \mathbb{E}\left(M_j\left|\hat{t}^N-t^*\right|\cdot 1_{\left\{\left|\hat{t}^N-t^*\right|<\delta\right\}}\right)\label{eq:uselipschitz}\\
&=& M_j \int^{\infty}_0 P\left( \left|\hat{t}^N-t^*\right|\cdot 1_{\left\{\left|\hat{t}^N-t^*\right|<\delta\right\}} > w\right)dw\nonumber\\
&\leq & M_j \int^\delta_0 P\left(\left|\hat{t}^N-t^*\right| > w\right)dw\nonumber\\
&\leq & M_j \int^\delta_0 C_4 e^{-C_4 w^2 N}dw\label{eq:againapplyconc}\\
&=& N^{-\frac{1}{2}} M_j C_3 \int^{\delta\sqrt{N}}_0 e^{-C_4 u^2}du,\label{eq:changeofvars}
\end{eqnarray}
where $M_j > 0$ in (\ref{eq:uselipschitz}) is obtained from the assumed upper bound on the density of $c\left(X,j\right)$, (\ref{eq:againapplyconc}) applies Theorem \ref{thm:conct}, and (\ref{eq:changeofvars}) uses a change of variables $u = w\sqrt{N}$. Note that $\int^\infty_0 e^{-C_4 u^2}du < \infty$. This completes the proof, because there are finitely many $j$ and $M_j \leq \max_k M_k$.

\subsection{Proof of Theorem \ref{thm:gradient}}

As in other proofs, let $\mathbb{E}^N$ denote the expectation over the empirical distribution, under which $P\left(X = X^n\right) = \frac{1}{N}$ for $n = 1,...,N$ and $X$ being independent of $X^1,...,X^N$. For $k=1,...,K-1$, define
\begin{equation*}
\rho_k\left(x,\theta,t,\psi\right) = \frac{e^{\theta_k}1_{\left\{k\in S\left(x,t,\psi\right)\right\}}}{1_{\left\{K\in S\left(x,t,\psi\right)\right\}+\sum_{j<K-1} e^{\theta}_j1_{\left\{j\in S\left(x,t,\psi\right)\right\}}}}.
\end{equation*}
Denote $R^N_k\left(\theta,t,\psi\right) = \mathbb{E}^N\left(\rho_k\left(X,\theta,t,\psi\right)\right)$ and $R_k\left(\theta,t,\psi\right) = \mathbb{E}\left(\rho_k\left(X,\theta,t,\psi\right)\right)$.

By definition, $p_k = R_k\left(\theta^*,t^*,\psi^*\right)$ and $\ell^N_k\left(\hat{\theta},\hat{t}^N,\hat{\psi}^N\right) = p_k - R^N_k\left(\hat{\theta},\hat{t}^N,\hat{\psi}^N\right)$. Therefore,
\begin{equation*}
R^N\left(\hat{\theta},\hat{t}^N,\hat{\psi}^N\right) - R\left(\theta^*,t^*,\psi^*\right) = -\ell^N\left(\hat{\theta},\hat{t}^N,\hat{\psi}^N\right).
\end{equation*}
We will write
\begin{equation}\label{eq:threeparts}
-\ell^N\left(\hat{\theta},\hat{t}^N,\hat{\psi}^N\right) = J_1 + J_2 + J_3
\end{equation}
where
\begin{eqnarray*}
J_1 &=& R^N\left(\hat{\theta},\hat{t}^N,\hat{\psi}^N\right) - R^N\left(\hat{\theta},t^*,\psi^*\right),\\
J_2 &=& R^N\left(\hat{\theta},t^*,\psi^*\right) - R\left(\hat{\theta},t^*,\psi^*\right),\\
J_3 &=& R\left(\hat{\theta},t^*,\psi^*\right) - R\left(\theta^*,t^*,\psi^*\right),
\end{eqnarray*}
and bound each of the three terms separately.

\textit{Step 1}. In this step, we show $\mathbb{E}\left(\|J_1\|_2\right) \lesssim N^{-\frac{1}{2}}$. We first write
\begin{eqnarray}
\|J_1\|_2 &=& \|\mathbb{E}^N\left(\rho\left(X,\hat{\theta},\hat{t}^N,\hat{\psi}^N\right) - \rho\left(X,\hat{\theta},t^*,\psi^*\right)\|_2\right)\nonumber\\
&\leq & \mathbb{E}^N\left(\|\rho\left(X,\hat{\theta},\hat{t}^N,\hat{\psi}^N\right) - \rho\left(X,\hat{\theta},t^*,\psi^*\right)\|_2 \right)\nonumber\\
&=& \mathbb{E}^N\left(\|\rho\left(X,\hat{\theta},\hat{t}^N,\hat{\psi}^N\right) - \rho\left(X,\hat{\theta},t^*,\psi^*\right)\|_2 1_{\left\{S\left(\hat{t}^N,\hat{\psi}^N,X\right) \neq S\left(t^*,\psi^*,X\right)\right\}}\right)\label{eq:addindicator}\\
&\leq& \sqrt{K}\cdot \mathbb{E}^N\left(1_{\left\{S\left(\hat{t}^N,\hat{\psi}^N,X\right) \neq S\left(t^*,\psi^*,X\right)\right\}}\right),\label{eq:binaryrho}
\end{eqnarray}
where (\ref{eq:addindicator}) holds because $\rho\left(X,\hat{\theta},\hat{t}^N,\hat{\psi}^N\right) = \rho\left(X,\hat{\theta},t^*,\psi^*\right)$ when $S\left(\hat{t}^N,\hat{\psi}^N,X\right) = S\left(t^*,\psi^*,X\right)$, and (\ref{eq:binaryrho}) is due to the fact that $\rho_k\left(X,\theta,t,\psi\right) \in \left(0,1\right)$ for all $k = 1,...,K-1$. Taking expectations, we obtain
\begin{eqnarray*}
\mathbb{E}\left(\|J_1\|_2\right) &\leq & \sqrt{K}\cdot P\left(S\left(\hat{t}^N,\hat{\psi}^N,X\right) \neq S\left(t^*,\psi^*,X\right)\right)\\
&\lesssim & N^{-\frac{1}{2}},
\end{eqnarray*}
where the last line follows by Proposition \ref{prop:newargmax}.

\textit{Step 2}. In this step, we show $\mathbb{E}\left(\|J_2\|_2\right) \lesssim N^{-\frac{1}{2}}$. We write
\begin{eqnarray}
\|J_2\|_2 &=& \sqrt{\sum_k \left(\mathbb{E}^N\left(\rho_k\left(X,\hat{\theta},t^*,\psi^*\right)\right) - \mathbb{E}\left(\rho_k\left(X,\hat{\theta},t^*,\psi^*\right)\right)\right)^2}\nonumber\\
&\leq & N^{-\frac{1}{2}}\sqrt{\sum_k \left(\Delta^N_k\left(\hat{\theta}\right)\right)^2}\nonumber\\
&\leq & N^{-\frac{1}{2}}\max_k \left|\Delta^N_k\left(\hat{\theta}\right)\right|,\label{eq:intermsofDelta}
\end{eqnarray}
where
\begin{eqnarray*}
\Delta^N_k\left(\theta\right) &=& \sqrt{N}\left(\mathbb{E}^N\left(\rho_k\left(X,\theta,t^*,\psi^*\right)\right) - \mathbb{E}\left(\rho_k\left(X,\theta,t^*,\psi^*\right)\right)\right)\\
&=& \sqrt{N}\sum_{T\subseteq\left\{1,...,K\right\},T\neq\emptyset} \left(P^N\left(S\left(X,t^*,\psi^*\right)=T\right) - P\left(S\left(X,t^*,\psi^*\right)=T\right)\right)\frac{e^{\theta_k}1_{\left\{k\in T\right\}}}{1_{\left\{K\in T\right\}} + \sum_{j<K-1}e^{\theta_j}1_{\left\{j\in T\right\}}}.
\end{eqnarray*}
Then, for any $\theta$,
\begin{equation*}
\left|\Delta^N_k\left(\theta\right)\right| \leq \sqrt{N} \sum_{T\subseteq\left\{1,...,K\right\},T\neq\emptyset} \left|P^N\left(S\left(X,t^*,\psi^*\right)=T\right) - P\left(S\left(X,t^*,\psi^*\right)=T\right)\right|.
\end{equation*}
By Hoeffding's inequality,
\begin{equation*}
P\left(\sqrt{N}\left|P^N\left(S\left(X,t^*,\psi^*\right)=T\right) - P\left(S\left(X,t^*,\psi^*\right)=T\right)\right| \geq w\right) \leq 2e^{-2w^2}.
\end{equation*}
A simple union bound yields
\begin{eqnarray*}
P\left(\sup_{\theta}\left|\Delta^N_k\left(\theta\right)\right| > 2^K w\right) \hspace{-0.05in}&\leq &\hspace{-0.05in} P\left(\sqrt{N}\sum_{T\subseteq\left\{1,...,K\right\},T\neq\emptyset} \left|P^N\left(S\left(X,t^*,\psi^*\right)=T\right) - P\left(S\left(X,t^*,\psi^*\right)=T\right)\right| > 2^K w\right)\\
\hspace{-0.05in}&\leq & \hspace{-0.05in}\sum_{T\subseteq\left\{1,...,K\right\},T\neq\emptyset} P\left(\sqrt{N}\left|P^N\left(S\left(X,t^*,\psi^*\right)=T\right) - P\left(S\left(X,t^*,\psi^*\right)=T\right)\right| \geq w\right)\\
\hspace{-0.05in}&\leq & \hspace{-0.05in}2^{K+1}e^{-2w^2}
\end{eqnarray*}
for any $k$. Therefore, we may write
\begin{equation*}
P\left(\max_k \left|\Delta^N_k\left(\hat{\theta}\right)\right| > w\right) \leq D e^{-D' w^2}
\end{equation*}
for some $D,D' > 0$. Returning to (\ref{eq:intermsofDelta}), we now have a concentration bound on $\|J_2\|_2$, whence it is straightforward to show that $\mathbb{E}\left(\|J_2\|_2\right) \lesssim N^{-\frac{1}{2}}$.

\textit{Step 3}. We now complete the proof. Combining the results of Steps 1-2 with (\ref{eq:threeparts}) and the assumption on $\ell^N\left(\hat{\theta},\hat{t}^N,\hat{\psi}^N\right)$, we obtain $\mathbb{E}\left(\|J_3\|_2\right) \lesssim N^{-\frac{1}{2}}$. By the fundamental theorem of calculus, we have $J_3 = -H^N\cdot\left(\hat{\theta}-\theta^*\right)$ where
\begin{equation}\label{eq:inthessian}
H^N = -\int^1_0 \frac{\partial^2 L\left(\theta^*+w\left(\hat{\theta}-\theta^*\right),t^*,\psi^*\right)}{\partial\theta\partial\theta^\top}dw
\end{equation}
and
\begin{equation*}
L\left(\theta,t,\psi\right) = p^\top \theta - \mathbb{E}\left(\log \left(\sum_{k\in S\left(t,\psi,X\right)} e^{\theta_k}\right)\right).
\end{equation*}
Because the mapping $\theta\mapsto L\left(\theta,t,\psi\right)$ is concave, we know that its Hessian is negative semidefinite. Therefore, taking into account the minus sign in (\ref{eq:inthessian}), $H^N$ is positive semidefinite. Note that $\theta\mapsto L\left(\theta,t,\psi\right)$ is continuous. Therefore, if we show that the largest eigenvalue of its Hessian at $\theta = \theta^*$ is bounded above by a strictly negative number, then we will have some $\kappa > 0$ such that $\lambda_{\min}\left(H^N\right)\geq \kappa$. This implies
\begin{equation*}
\|\hat{\theta}-\theta^*\|_2 = \|\left(H^N\right)^{-1}J_3\|_2 \leq \frac{1}{\kappa}\|J_3\|_2\lesssim N^{-\frac{1}{2}}.
\end{equation*}
Thus, to complete the proof, it remains only to establish the following technical lemma, whose proof is given in a separate subsection of the Appendix.

\begin{lem}\label{lem:hessian}
Let $H\left(\theta\right) = -\frac{\partial^2 L\left(\theta,t^*,\psi^*\right)}{\partial\theta\partial\theta^\top}$. Let $\theta^*$ be the maximizer of $\theta\mapsto L\left(\theta,t^*,\psi^*\right)$. Then, $\theta^*$ is unique and
\begin{equation*}
\lambda_{\min}\left(H\left(\theta^*\right)\right) \geq \left(\min_{j<K} p_j\right)p_K.
\end{equation*}
\end{lem}

\subsection{Proof of Theorem \ref{thm:final}}

Observe that
\begin{eqnarray*}
\left|\rho_k\left(X,\hat{\theta},\hat{t}^N,\hat{\psi}^N\right)-\rho_k\left(X,\theta^*,t^*,\psi^*\right)\right| &\leq & \left|\rho_k\left(X,\hat{\theta},\hat{t}^N,\hat{\psi}^N\right)-\rho_k\left(X,\theta^*,t^*,\psi^*\right)\right|\cdot 1_{\left\{S\left(X,\hat{t}^N,\hat{\psi}^N\right) = S\left(X,t^*,\psi^*\right)\right\}}\\
&\,& + 1_{\left\{S\left(X,\hat{t}^N,\hat{\psi}^N\right) \neq S\left(X,t^*,\psi^*\right)\right\}}.
\end{eqnarray*}
It can be verified by direct computation that
\begin{equation*}
\left|\frac{\partial \rho_k\left(x,\theta,t,\psi\right)}{\partial \theta_j}\right| \leq 1, \qquad \forall j,k.
\end{equation*}
Therefore, on the event $\left\{S\left(X,\hat{t}^N,\hat{\psi}^N\right) = S\left(X,t^*,\psi^*\right)\right\}$, we have
\begin{equation*}
\left|\rho_k\left(X,\hat{\theta},\hat{t}^N,\hat{\psi}^N\right)-\rho_k\left(X,\theta^*,t^*,\psi^*\right)\right| \leq \sqrt{K}\cdot \|\hat{\theta}-\theta^*\|_2.
\end{equation*}
Consequently,
\begin{equation*}
\mathbb{E}\left(\left|\rho_k\left(X,\hat{\theta},\hat{t}^N,\hat{\psi}^N\right)-\rho_k\left(X,\theta^*,t^*,\psi^*\right)\right|\right) \leq \sqrt{K}\cdot \mathbb{E}\left(\|\hat{\theta}-\theta^*\|_2\right) + P\left(S\left(X,\hat{t}^N,\hat{\psi}^N\right) \neq S\left(X,t^*,\psi^*\right)\right).
\end{equation*}
The desired result then follows by Proposition \ref{prop:newargmax} and Theorem \ref{thm:gradient}.

\subsection{Proof of Lemma \ref{lem:technical}}

Fix $\psi \in\left\{0,1\right\}^K$ and consider $\Delta^N\left(\psi,t\right)$ as a function of $t$. Define
\begin{eqnarray*}
\mathcal{F}_k &=& \left\{1_{\left\{c\left(X,k\right)\leq t\right\}} + \psi_k: t\in\mathbb{R}\right\},\\
\mathcal{F} &=& \left\{\max_k 1_{\left\{c\left(X,k\right)\leq t\right\}} + \psi_k: t\in\mathbb{R}\right\}.
\end{eqnarray*}
Then, $\mathcal{F}_k$ is a VC-subgraph class (see, e.g., Lem. 2.6.16 of \citealp{VaWe96}). By Lem. 2.6.18 of \cite{VaWe96}, $\bigvee_k \mathcal{F}_k$ is a VC-subgraph class as well. Since $\mathcal{F} \subseteq \bigvee_k \mathcal{F}_k$, it follows that $\mathcal{F}$ is a VC-subgraph class. Note that every element of $\mathcal{F}$ is bounded by $2$. Then, by Thm. 2.6.7 and Thm. 2.14.9 of \cite{VaWe96}, there exist $\kappa_1,\kappa_2 > 0$ such that
\begin{equation}\label{eq:technical1}
P\left(\sup_t \left|\Delta^N\left(\psi,t\right)\right| > w\right) \leq \kappa_1 w^{\kappa_2 r} e^{-2w^2}
\end{equation}
for any $r \geq 1$ and $w > 0$. We may rewrite (\ref{eq:technical1}) as
\begin{equation*}
P\left(\sup_t \left|\Delta^N\left(\psi,t\right)\right| > w\right) \leq \kappa_3 e^{-\kappa_4 w^2}
\end{equation*}
for suitable $\kappa_3,\kappa_4 > 0$. By the union bound,
\begin{equation*}
P\left(\sup_{\psi\in\left\{0,1\right\}^K} \sup_t \left|\Delta^N\left(\psi,t\right)\right| > w\right) \leq 2^K \kappa_3 e^{-\kappa_4 w^2},
\end{equation*}
which completes the proof.

\subsection{Proof of Lemma \ref{lem:technical2}}

By assumption, the density $\xi$ of $\left\{c\left(X,k\right)\right\}^K_{k=1}$ is bounded below by $C>0$ on the set $T = \left\{z\in\mathbb{R}^K:\|z - \left(t^*,...,t^*\right)\|_{\infty} \leq C'\right\}$ for some sufficiently small $C' > 0$.

\textit{Step 1}. First, we show that a similar lower bound holds for the lower-dimensional random vector $\left\{c\left(X,k\right)\right\}_{k\in A}$ where $A\subseteq\left\{1,...,K\right\}$ is nonempty. Without loss of generality, let $A = \left\{1,...,k'\right\}$ for some $1 \geq k' < K$. Then, the density $\xi_A$ of $\left\{c\left(X,k\right)\right\}_{k\in A}$ can be written as
\begin{eqnarray*}
\xi_A\left(z_1,...,z_r\right) &=& \int\ldots\int \xi\left(z_1,...,z_K\right)dz_{r+1}\ldots dz_K\\
&\geq & \int\ldots\int 1_{\left\{z_1,...,z_K\in T\right\}}\xi\left(z_1,...,z_K\right)dz_{r+1}\ldots dz_K\\
&\geq & C\int\ldots\int 1_{\left\{z_1,...,z_K\in T\right\}}dz_{r+1}\ldots dz_K\\
&=& C\left(2C'\right)^{K-r}.
\end{eqnarray*}

\textit{Step 2}. Next, we show that, for any nonempty $A\subseteq\left\{1,...,K\right\}$,
\begin{equation}\label{eq:techinterval}
P\left(\left\{\forall k\in A,\, c\left(X,k\right) > t^*\right\}\cap\left\{\exists k\in A : c\left(X,k\right) < t^*+\delta\right\}\right) \geq \bar{C}_{\left|A\right|}\delta,
\end{equation}
where $\bar{C}_{\left|A\right|}$ depends only on $\left|A\right|$. From Step 1, we know that (\ref{eq:techinterval} holds if $\left|A\right|=1$. Now suppose that $\left|A\right|>1$, again writing $A=\left\{1,...,k'\right\}$ for some $k'>1$. Then,
\begin{eqnarray}
&\,& P\left(\left\{\forall k\leq k',\, c\left(X,k\right) > t^*\right\}\cap\left\{\exists k\leq k' : c\left(X,k\right) < t^*+\delta\right\}\right)\nonumber\\
&\geq & P\left(\left\{\forall k < k',\, c\left(X,k\right) > t^*\right\}\cap\left\{t^* < c\left(X,k'\right) < t^*+\delta\right\}\right)\nonumber\\
&\geq & \int\ldots\int 1_{\left\{c\left(X,1\right) > t^*\right\}}\cdot ...\cdot 1_{\left\{c\left(X,k'\right) > t^*\right\}}\cdot 1_{\left\{c\left(X,k'\right) < t^*+\delta\right\}}\xi_{\left\{1,...,k'\right\}}\left(z_1,...,z_{k'}\right)dz_1\ldots dz_{k'}\nonumber\\
&\geq & \int\ldots\int 1_{\left\{t^* < c\left(X,1\right) < t^*+\frac{1}{2}C'\right\}}\cdot ...\cdot 1_{\left\{t^* < c\left(X,k'-1\right) < t^*+\frac{1}{2}C'\right\}}\cdot 1_{\left\{t^* < c\left(X,k'\right) < t^*+\delta\right\}}\xi_{\left\{1,...,k'\right\}}\left(z_1,...,z_{k'}\right)dz_1\ldots dz_{k'}\nonumber\\
&\geq & C\left(2C'\right)^{K-k'}\int\ldots\int 1_{\left\{t^* < c\left(X,1\right) < t^*+\frac{1}{2}C'\right\}}\cdot ...\cdot 1_{\left\{t^* < c\left(X,k'-1\right) < t^*+\frac{1}{2}C'\right\}}\cdot 1_{\left\{t^* < c\left(X,k'\right) < t^*+\delta\right\}}dz_1\ldots dz_{k'}\label{eq:applystep1}\\
&\geq & C\left(2C'\right)^{K-k'}\left(\frac{1}{2}C'\right)^{k'-1}\delta,\nonumber
\end{eqnarray}
where (\ref{eq:applystep1}) applies the result of Step 1.

\textit{Step 3}. In this step, we show the existence of $D>0$ such that
\begin{equation*}
\mathbb{E}\left(\max_k 1_{c\left(X,k\right)\leq t^*+\delta} + \psi_k\right) \geq D\cdot\delta + \mathbb{E}\left(\max_k 1_{c\left(X,k\right)\leq t^*} + \psi_k\right), \qquad \forall \psi\in\left\{0,1\right\}^K.
\end{equation*}
Fix some $\psi\in\left\{0,1\right\}^K$ and let $A = \left\{k:\psi_k=1\right\}$. Suppose that $A\neq\emptyset$. By repeating the arguments in the proof of Proposition \ref{prop:monotonicity}, we obtain
\begin{eqnarray*}
&\,& \mathbb{E}\left(\max_k 1_{\left\{c\left(X,k\right)\leq t^*+\delta\right\}} + \psi_k\right) - \mathbb{E}\left(\max_k 1_{\left\{c\left(X,k\right)\leq t^*\right\}} + \psi_k\right)\\
&=& P\left(c\left(X,k\right) > t^*,\,\forall k\in A\right) - P\left(c\left(X,k\right) > t^*+\delta,\,\forall k\in A\right)\\
&\geq & \bar{C}_{\left|A\right|}\delta,
\end{eqnarray*}
where the last line applies (\ref{eq:techinterval}) proved in Step 2. On the other hand, if $A=\emptyset$, we may again repeat the arguments in the proof of Proposition \ref{prop:monotonicity} and obtain
\begin{eqnarray*}
&\,& \mathbb{E}\left(\max_k 1_{\left\{c\left(X,k\right)\leq t^*+\delta\right\}} + \psi_k\right) - \mathbb{E}\left(\max_k 1_{\left\{c\left(X,k\right)\leq t^*\right\}} + \psi_k\right)\\
&=& P\left(c\left(X,k\right) > t^*,\;\forall k\right) - P\left(c\left(X,k\right) > t^*+\delta,\;\forall k\right)\\
&\geq & \bar{C}_{K}\delta,
\end{eqnarray*}
where again the last line is due to (\ref{eq:techinterval}). We may take $D = \min\left\{\bar{C}_1,...,\bar{C}_K\right\}$ to complete this step.

\textit{Step 4}. Take
\begin{equation*}
\psi'\in\arg\min_{\psi\in\left\{0,1\right\}^K} \mathbb{E}\left(\max_k 1_{\left\{c\left(X,k\right)\leq t^*+\delta\right\}} + \psi_k\right) - p^\top \psi.
\end{equation*}
Then,
\begin{eqnarray}
\min_{\psi\in\left\{0,1\right\}^K} \mathbb{E}\left(\max_k 1_{\left\{c\left(X,k\right)\leq t^*\right\}}+\psi_k\right)-p^\top\psi +D\delta &\leq & \mathbb{E}\left(\max_k 1_{\left\{c\left(X,k\right)\leq t^*\right\}}+\psi'_k\right)-p^\top\psi' +D\delta\nonumber\\
&\leq& \mathbb{E}\left(\max_k 1_{\left\{c\left(X,k\right)\leq t^*+\delta\right\}}+\psi'_k\right)-p^\top\psi'\label{eq:applystep3}\\
&=& \min_{\psi\in\left\{0,1\right\}^K} \mathbb{E}\left(\max_k 1_{\left\{c\left(X,k\right)\leq t^*+\delta\right\}}+\psi_k\right)-p^\top\psi,\nonumber
\end{eqnarray}
where (\ref{eq:applystep3}) applies the result of Step 3. This completes the proof.

\subsection{Proof of Lemma \ref{lem:hessian}}

By direct computation, one may derive $H\left(\theta\right) = \text{diag}\left(\sigma\left(\theta\right)\right) - \sigma\left(\theta\right)\sigma\left(\theta\right)^\top$, where $\sigma\left(\theta\right)$ is a $\left(K-1\right)$-vector with
\begin{equation*}
\sigma_k\left(\theta\right) = \mathbb{E}\left(\frac{e^{\theta_k}1_{\left\{k\in S\left(t^*,\psi^*,X\right)\right\}}}{1_{\left\{K\in S\left(t^*,\psi^*,X\right) + \sum_{j<K} e^{\theta_j} 1_{\left\{j\in S\left(t^*,\psi^*,X\right)\right\}}\right\}}}\right).
\end{equation*}
Similarly, define
\begin{equation*}
\sigma_K\left(\theta\right) = \mathbb{E}\left(\frac{1_{\left\{K\in S\left(t^*,\psi^*,X\right)\right\}}}{1_{\left\{K\in S\left(t^*,\psi^*,X\right) + \sum_{j<K} e^{\theta_j} 1_{\left\{j\in S\left(t^*,\psi^*,X\right)\right\}}\right\}}}\right).
\end{equation*}
For any $\lambda \in\mathbb{R}^{K-1}$ with $\sum_j \lambda^2_j = 1$, we have
\begin{eqnarray}
\lambda^\top H\left(\theta\right)\lambda &=& \sum_{j<K} \lambda^2_j\sigma_j\left(\theta\right) - \left(\sum_{j<K} \lambda_j \sigma_j\left(\theta\right)\right)^2\nonumber\\
&\geq & \sum_{j < K}\lambda^2_j\sigma_j\left(\theta\right) - \left(\sum_{j<K} \lambda^2_j\sigma_j\left(\theta\right)\right)\left(\sum_{j<K}\sigma_j\left(\theta\right)\right)\label{eq:hessian1}\\
&=& \left(\sum_{j<K}\lambda^2_j\sigma_j\left(\theta\right)\right)\left(1 - \sum_{j<K}\sigma_j\left(\theta\right)\right)\nonumber\\
&\geq & \left(\min_{j<K}\sigma_j\left(\theta\right)\right)\sigma_K\left(\theta\right),
\end{eqnarray}
where (\ref{eq:hessian1}) follows by the Cauchy-Schwarz inequality. Because $\theta^*$ maximizes $\theta\mapsto L\left(\theta,t^*,\psi^*\right)$, we have $\sigma_j\left(\theta^*\right) = p_j$ for $j = 1,...,K$. Therefore,
\begin{equation*}
\lambda^\top H\left(\theta^*\right)\lambda \geq \left(\min_{j<K}p_j\right)p_K,
\end{equation*}
as required. The uniqueness of $\theta^*$ follows by the fact that $H\left(\theta^*\right)$ is positive definite and $\theta\mapsto H\left(\theta\right)$ is continuous.

\end{document}

%% file: notation_ams.tex
\def \qed{\hfill $\Box$}

\def \notin{{\not \in}}

\def \hbar{\bar h}

\def \v1 {\mathbb{F}}



%% file: macro_ams.tex
\newcommand{\doublespace}{\addtolength{\baselineskip}{.5\baselineskip}}
\newcounter{stepno}

\newtheorem{thm}{Theorem}[section]

\newtheorem{lem}{Lemma}[section]

\newtheorem{corol}{Corollary}[section]

\newtheorem{prop}{Proposition}[section]

\newcommand{\bnn}{\\ \\$\begin{array}{rcll}}
\newcommand{\enn}{\end{array}$\\ \\}
\newcommand{\ba}{\begin{eqnarray}}
\newcommand{\ea}{\end{eqnarray}}
\newcommand{\bas}{\begin{eqnarray*}}
\newcommand{\eas}{\end{eqnarray*}}
\newcommand{\bdm}{\begin{displaymath}}
\newcommand{\edm}{\end{displaymath}}
\newcommand{\be}{\begin{equation}}
\newcommand{\ee}{\end{equation}}
\newcommand{\bn}{\begin{eqnarray}}
\newcommand{\en}{\end{eqnarray}}
\newcommand{\bns}{\begin{eqnarray*}}
\newcommand{\ens}{\end{eqnarray*}}

\newcommand{\defvarbegin}{\begin{quotation}\vspace{-15pt}\begin{tabbing}}
\newcommand{\defvarend}  {\end{tabbing}\vspace{-10pt}\end{quotation}}

\newcommand{\bnarray}{\begin{equation}\begin{array}{rcll}}
\newcommand{\enarray}{\end{array}\end{equation}}

\newcommand{\barr}{\begin{array}}
\newcommand{\earr}{\end{array}}

\newcounter{cnum}

\newcommand{\beginalg}{\setcounter{stepno}{1}
                \begin{list}{\bf Step~\arabic{stepno}}
                         {\usecounter{stepno}\settowidth{\labelwidth}{\bf Step~9m}
                \addtolength{\leftmargin}{2\parindent}}
                }
\newcommand{\eg}{\end{list}}

\def \define{\begin{quote}\begin{itemize}}
\def \enddefine{\end{itemize}\end{quote}}

\newlength{\boxedparwidth} 
  {\begin{center} \begin{tabular}{|@{\hspace{.15in}}c@{\hspace{.15in}}|}
                 \hline \\ \begin{minipage}[t]{\boxedparwidth}
                 \setlength{\parindent}{0.0in}}%
   {\end{minipage} \\ \\ \hline \end{tabular} \end{center}}
\newcounter{chapterctr}
\newcounter{example}[chapterctr]
\newcounter{ctr}